\documentclass[showpacs, amsmath, amssymb, onecolumn, notitlepage]{revtex4-2}

\usepackage{appendix}
\usepackage{bbm}
\usepackage{braket}
\usepackage{amsfonts}
\usepackage{amsthm}
\usepackage{amsmath}

\usepackage{txfonts}
\usepackage{bm}
\usepackage[dvipdfmx]{graphicx}
\usepackage{color}
\usepackage{mathrsfs}
\usepackage{subfigure}
\usepackage{wrapfig}
\usepackage{here}
\usepackage{hyperref}
\usepackage{stackengine}
\usepackage{cellspace}
\newcommand{\be}{\begin{equation}}
\newcommand{\ee}{\end{equation}}
\newcommand{\ba}{\begin{align}}
\newcommand{\ea}{\end{align}}

\newtheorem{thm}{Theorem}
\newtheorem{lemma}{Lemma}

\begin{document}

\title{Security of Quantum Conference Key Agreement with Two-Way Classical Communication
}

\author{Shun Kawakami$^{1,2}$, Mori Watanabe$^{3}$, Takuya Ikuta$^{4}$, Koichi Takasugi$^{1}$
}

\affiliation{
$^1$ Network Innovation Laboratories, NTT Inc.,\\ 1-1 Hikari-no-oka, Yokosuka, Kanagawa, 239-0847, Japan \\
$^2$ NTT Research Center for Theoretical Quantum Information, NTT Inc.,\\ \mbox{3-1 Morinosato Wakamiya, Atsugi, Kanagawa,~243-0198, Japan}\\
$^3$ Department of Advanced Materials Science, Graduate School of Frontier Sciences, The University of Tokyo, \mbox{5-1-5 Kashiwanoha, Kashiwa, Chiba, 277-0882, Japan}\\
$^4$ Basic Research Laboratories, NTT Inc.,\\ \mbox{3-1 Morinosato Wakamiya, Atsugi, Kanagawa,~243-0198, Japan}\\
}

\begin{abstract}
Quantum conference key agreement (QCKA) enables multiple users to establish a common secret key with information-theoretic security and is regarded as a key primitive for secure communication in future quantum networks.
However, practical implementations of QCKA typically suffer from higher noise levels than conventional bipartite quantum key distribution (QKD), making the improvement of the tolerable error threshold an important challenge.
Gottesman and Lo proposed two preprocessing procedures for QKD with two-way classical communication, known as the B-step and the P-step, which enhance the tolerable error threshold.
In this paper, we analyze the asymptotic security of QCKA with tripartite GHZ states and two measurement bases using two-way classical communication, including multiple B-steps and P-steps.
We derive the corresponding secure key rate analytically and demonstrate that iterative B-steps can increase the tolerable error threshold beyond $20\%$, significantly improving upon the approximately $11\%$ threshold achievable without two-way classical communication and the approximately $15\%$ threshold obtained with only a single B-step.
Our results show that two-way classical communication can substantially enhance the robustness of practical QCKA protocols.

\end{abstract}

\maketitle

\section{Introduction}
Quantum networks are expected to enable distributed quantum information processing among geographically separated users, providing a foundation for applications such as secure communications, distributed quantum computing, and quantum sensing~\cite{2008Kimble, 2018Wehner}. Among their cryptographic applications, quantum conference key agreement (QCKA) allows multiple users to establish a common secret key with information-theoretic security. As a multipartite extension of quantum key distribution (QKD),
which is one of the most advanced applications in quantum communication,
QCKA is regarded as a key primitive for secure communication in future quantum networks.

Among various approaches to QCKA, protocols based on multipartite Greenberger-Horne-Zeilinger (GHZ) states~\cite{1989GHZ} have attracted particular attention because GHZ states naturally provide the global correlations required for conference key generation. The security of GHZ-based QCKA has been established theoretically \cite{2018Grasselli}, and proof-of-principle demonstrations have been reported~\cite{2021Proietti, 2023Pickston, 2026Zou}. However, current experimental implementations typically exhibit higher noise levels than conventional bipartite QKD systems. Consequently, improving the tolerable error threshold is an important challenge for the practical deployment of QCKA.

In QKD, one of the most successful approaches to increasing the tolerable error threshold is the use of two-way classical communication. Gottesman and Lo proposed and analyzed two preprocessing procedures on sifted keys, known as the B-step and the P-step~\cite{2003Gottesman}. The B-step suppresses bit errors through a process closely related to classical advantage distillation~\cite{1993Maurer}, while the P-step reduces phase errors. Their analysis showed that the error threshold for obtaining a positive key rate can be increased from approximately 11\% to 19\% for the Bennett-Brassard 1984 (BB84) protocol~\cite{1984Bennett} and the Bennett-Brassard-Mermin 1992 (BBM92)  protocol~\cite{1992Bennett},  both of which are representative QKD protocols.

The two-way classical communication was first applied to QCKA using tripartite GHZ states in \cite{2007Chen}. In their analysis, three measurement bases are required to fully characterize the GHZ-diagonal state. From a practical perspective, however, protocols using only two measurement bases are more desirable, similarly to BB84 and BBM92 protocols. The difficulty is that, with only two bases, the GHZ-diagonal parameters are not uniquely determined, making the security analysis more challenging. This gap between theory and practical implementation remained unresolved for many years.

Recently, this problem was addressed in \cite{2025Krawec} and subsequently improved in \cite{2026Krawec}. However, their analysis is restricted to a single B-step, yielding only a limited improvement in the tolerable error threshold. Moreover, the improved key-rate formula in \cite{2026Krawec} is expressed through an optimization problem whose solution is only obtained numerically. 

In this paper, we derive an asymptotic secure key rate for QCKA with tripartite GHZ states and two measurement bases using two-way classical communication, including multiple B-steps and P-steps. In particular, we derive an analytical solution to the optimization problem associated with a single B-step and show that the optimization problem after multiple B-step iterations can be reduced to a recursive sequence of single-step optimizations. This avoids the need to directly optimize the increasingly complicated objective function arising from repeated B-step applications and allows the resulting optimization problem to be solved analytically for any number of B-steps.  Using this result, we demonstrate that the tolerable error threshold can exceed 20\%, representing a substantial improvement over the approximately 11\% threshold achievable without two-way classical communication and the approximately 15\% threshold obtained with only a single B-step. Our results imply that two-way classical communication can substantially improve the robustness of practical GHZ-based QCKA protocols using only two measurement bases. 

This paper is organized as follows. In Sec.~\ref{Sec:protocol}, we describe the QCKA protocol with two-way classical communication and the assumptions used in our security analysis. In Sec.~\ref{Sec: Preliminary}, we formulate the optimization problem for deriving the secure key rate under general B-steps and P-steps. In Sec.~\ref{Sec:singleB}, we analytically solve the optimization problem for a single B-step and identify the solution relevant to positive key generation. In Sec.~\ref{Sec:multiB}, we extend the analysis to multiple B-steps and derive the corresponding secure key rate. In Sec.~\ref{Sec:multiP}, we analyze multiple P-steps and their combination with B-steps. Numerical results demonstrating the improvement of the tolerable error threshold are presented in Sec.~\ref{Sec:Numerical}. Finally, Sec.~\ref{Sec:Conclusion} concludes the paper.

\section{Protocol and assumptions} \label{Sec:protocol}
Here, we introduce a QCKA protocol with a GHZ state using two-way classical communication, in which three legitimate parties, Alice, Bob and Charlie share a secret key. 
In this protocol, two complementary bases ($Z$ and $X$) are used to prepare non-orthogonal states, where the secret key is extracted from the $Z$ basis and the signal disturbance is monitored by using the $X$ basis. The protocol is described as follows. \\

\noindent
(1) {\it State Preparation and distribution}:  
For each round $l \in \{1, \dots, N\}$, a tripartite entangled state is prepared as
\begin{equation}
\label{idealGHZstate}
    \ket{\Psi}_{ABC} = \frac{1}{\sqrt{2}}(\ket{0}_A\ket{0}_B\ket{0}_C + \ket{1}_A\ket{1}_B\ket{1}_C).
\end{equation}
Then, its three subsystems are distributed to Alice, Bob, and Charlie via quantum channels.
\\
(2) {\it Measurement}: 
Alice selects the basis $W_A^{(l)} \in \{Z,X\}$, and accordingly performs the $Z~(X)$-basis measurement $\{\ket{0}\bra{0}_A, \ket{1}\bra{1}_A\}$ ($\{\ket{+}\bra{+}_A, \ket{-}\bra{-}_A\}$) on the qubit to obtain an outcome $b_A^{(l)} \in \{0,1\}$. Here, $\ket{+}_A\coloneqq (\ket{0}_A + \ket{1}_A)/\sqrt{2}$ and $\ket{-}_A\coloneqq (\ket{0}_A - \ket{1}_A)/\sqrt{2}$, corresponding to the outcome 0 and 1, respectively. 
In the same manner, Bob (Charlie) selects the basis $W_B^{(l)}$ ($W_C^{(l)}$) and obtains the outcome $b_B^{(l)}$ ($b_C^{(l)}$). 
\\
(3) {\it Repetition}: Alice, Bob, Charlie repeat steps (1) and (2) $N$ times. 
\\
(4) {\it Public communication}: For each round $l$, Alice, Bob and Charlie publicly announce $W_A^{(l)}$, $W_B^{(l)}$ and $W_C^{(l)}$. 
They keep only the rounds satisfying $W_A^{(l)} = W_B^{(l)} = W_C^{(l)} =Z$ or $W_A^{(l)} = W_B^{(l)} = W_C^{(l)} =X$, and all other rounds are discarded.  
For rounds with $W_A^{(l)} = W_B^{(l)} = W_C^{(l)} =Z$, bits in partially sampled rounds are disclosed. 
For rounds with $W_A^{(l)} = W_B^{(l)} = W_C^{(l)} =X$, all bits are disclosed. 
\\
(5) {\it Parameter estimation}:
By using the disclosed bits, the parties estimate the phase error rate 
\be
\label{defs1}
s_1 = {\rm Pr}(b_A^{(l)} \oplus b_B^{(l)} \oplus b_C^{(l)} =1~|~W_A^{(l)} = W_B^{(l)} = W_C^{(l)} =X), 
\ee
where `$\oplus$' represents exclusive OR.  
They also estimate the bit error rate 
\be
\begin{split}
\label{defs2s3}
s_2 = {\rm Pr}(b_A^{(l)} \neq b_B^{(l)} |~W_A^{(l)} = W_B^{(l)} = W_C^{(l)} = Z), \\
s_3 = {\rm Pr} (b_A^{(l)} \neq b_C^{(l)} |~W_A^{(l)} = W_B^{(l)} = W_C^{(l)} = Z), \\
s_4 = {\rm Pr} (b_B^{(l)} \neq b_C^{(l)} |~W_A^{(l)} = W_B^{(l)} = W_C^{(l)} = Z).
\end{split}
\ee
(6) {\it Two-way classical communications}:
Define an asymptotic key rate $R$ as \cite{2007Chen} 
\be
\label{initialkeyrate}
R \coloneqq 1 - f\cdot\max\{h(s_2),h(s_3)\} -h(s_1),
\ee
where $f~(\geq 1)$ denotes the efficiency of error correction and $h(x) = -x \log(x) - (1-x) \log(1-x)$ is the binary entropy. 
If $R > 0$, this step is skipped. Otherwise, the parties perform classical processing B-step or P-step~\cite{2007Chen}, corresponding to bit-error suppression and phase-error suppression, respectively.  
\begin{description}
\item[$\cdot$ B-step]
The parties randomly group the remaining $Z$-basis rounds into pairs. 
For each pair $(l,m)$:\\
1. Each party computes the parity of their two bits:
$p_A = b_A^{(l)} \oplus b_A^{(m)}$, $p_B = b_B^{(l)} \oplus b_B^{(m)}$, $p_C = b_C^{(l)} \oplus b_C^{(m)}$. \\
2. They publicly compare $(p_A, p_B, p_C)$. \\
3. The pair is kept only if $p_A = p_B = p_C$, otherwise it is discarded. \\
4. From each surviving pair, only one bit $(b_A^{(l)}, b_B^{(l)}, b_C^{(l)})$ is retained as the new key. 
\end{description}
\begin{description}
\item[$\cdot$ P-step]
The parties randomly group the remaining $Z$-basis rounds into blocks of size three. 
For each block $(l,m, n)$:\\
1. Each party computes the parity over the block: 
$\tilde{p}_A = b_A^{(l)} \oplus b_A^{(m)} \oplus b_A^{(n)}$, $\tilde{p}_B = b_B^{(l)} \oplus b_B^{(m)} \oplus b_B^{(n)}$, $\tilde{p}_C = b_C^{(l)} \oplus b_C^{(m)} \oplus b_C^{(n)}$. \\
2. The resulting bits $(\tilde{p}_A, \tilde{p}_B, \tilde{p}_C)$ are defined as the new key. \\
\end{description}
The B-step and P-step update the parameters $s_1$, $s_2$, $s_3$ and $s_4$, of which the expressions are shown in Sec.~\ref{Sec: Preliminary}. \\
(7) {\it Iteration of B- and P-steps }: 
Step (6) is iterated with updated $s_1$, $s_2$, $s_3$ and $s_4$. 
The iteration sequence is optimized depending on the original observed error rates $s_1$, $s_2$, $s_3$ and $s_4$. 
If any sequence is not expected to achieve a positive secure key rate $R>0$, the protocol aborts.
\\
(8) {\it Error correction}: Let $M$ denote the length of the remaining bit sequence. Alice generates a syndrome of length $Mf\cdot\max\{h(s_2),h(s_3)\}$ based on the updated bit error rate $s_2, s_3$ from her sifted key. 
She encrypts this syndrome using a pre-shared secret key with Bob and Charlie and sends it to them~\cite{2006Koashi}. Bob and Charlie then correct their keys according to the decrypted syndrome. 
\\
(9) {\it Privacy amplification}: With the updated phase error rate $s_1$, Alice, Bob and Charlie perform privacy amplification by shortening their keys by $Mh(s_1)$ bits to obtain the final secret key. 
\\\\

For this protocol, we derive a secure key rate in the asymptotic limit \(N \to \infty\) under the following assumptions:
\begin{equation}
s_1, s_2, s_3, s_4 < \frac{1}{2}.
\label{assumptionlessthanhalf}
\end{equation}
In addition, we assume that an eavesdropper Eve identically interacts with distributed qubit states in every round (IID assumption).

\section{Optimization problem in parameter updates}
\label{Sec: Preliminary}
In this section, we briefly introduce the notation and review how the B-step and P-step update the parameters in the GHZ basis, following~\cite{2007Chen}. Based on these expressions, we then state the main optimization problem addressed in this paper.
In the protocol described in the previous section, Alice, Bob and Charlie
wish to share the GHZ state, which has the form 
in Eq.~(\ref{idealGHZstate}). 
The GHZ state is the common eigenstate having eigenvalue of +1 for the following commuting observables,
\begin{equation}
    \begin{split}
        &\hat{S}_{1} = \hat{X} \otimes \hat{X} \otimes \hat{X},\\ 
        &\hat{S}_{2} = \hat{Z} \otimes \hat{Z} \otimes \hat{I},\\
        &\hat{S}_{3} = \hat{Z} \otimes \hat{I} \otimes \hat{Z},\\
    \end{split}
\end{equation}
where  $\hat{I} = \begin{pmatrix}
    1 & 0\\
    0 & 1\\
\end{pmatrix}$, $\hat{X} = \begin{pmatrix}
    0 & 1\\
    1 & 0\\
\end{pmatrix}$, $\hat{Y} = \begin{pmatrix}
    0 & -i\\
    i & 0\\
\end{pmatrix}$, $\hat{Z} = \begin{pmatrix}
    1 & 0\\
    0 & -1\\
\end{pmatrix}$.
In other words, the three observables are the stabilizer generators of the GHZ state.
Additionally, other nontrivial stabilizer elements,
\begin{equation}
    \label{Stabelements}
    \begin{split}
        &\hat{S}_{4} = \hat{I} \otimes \hat{Z} \otimes \hat{Z},\\
        &\hat{S}_{5} = -\hat{Y} \otimes \hat{Y} \otimes \hat{X},\\ 
        &\hat{S}_{6} = -\hat{Y} \otimes \hat{X} \otimes \hat{Y},\\
        &\hat{S}_{7} = -\hat{X} \otimes \hat{Y} \otimes \hat{Y},\\
    \end{split}
\end{equation}
can be obtained from multiplication of the three stabilizer generators.

The GHZ basis is defined by the eigenvalue of these three stabilizer generators,
\begin{equation}
   \ket{\Psi_{p,i_1,i_2}}_{ABC} = \frac{1}{\sqrt{2}}(\ket{0}_A\ket{i_1}_B\ket{i_2}_C + 
   (-1)^p  \ket{1}_A\ket{\overline{i_1}}_B\ket{\overline{i_2}}_C).
\end{equation}
where $p, i_1, i_2$ take values 0 or 1, respectively, and the bar over $i_1$ and $i_2$ indicate
its logical negation. The value 0 for $p, i_1, i_2$ corresponds to the eigenvalue of +1 for
$\hat{S}_1, \hat{S}_2, \hat{S}_3$, while the value 1 for $p, i_1, i_2$ corresponds to the eigenvalue of -1.
Then, we can label a density matrix which is diagonal in the GHZ basis,
\begin{equation}
\hat{\rho}_{ABC} =
\begin{pmatrix}
    p_{000} & 0 & 0 & 0 & 0 & 0 & 0 & 0 \\
    0 & p_{100} & 0 & 0 & 0 & 0 & 0 & 0 \\
    0 & 0 & p_{011} & 0 & 0 & 0 & 0 & 0 \\
    0 & 0 & 0 & p_{111} & 0 & 0 & 0 & 0 \\
    0 & 0 & 0 & 0 & p_{010} & 0 & 0 & 0 \\
    0 & 0 & 0 & 0 & 0 & p_{110} & 0 & 0 \\
    0 & 0 & 0 & 0 & 0 & 0 & p_{001} & 0 \\
    0 & 0 & 0 & 0 & 0 & 0 & 0 & p_{101} \\
\end{pmatrix},
\end{equation}
where $\{p_{p,i_1,i_2}\}$ are the probability amplitudes. 
The off-diagonal terms can be set to zero without loss of generality~\cite{2003Gottesman, 2025Du}, since inserting a GHZ-basis measurement prior to Alice’s, Bob’s, and Charlie’s measurements does not affect the statistics relevant to the virtual entanglement-purification protocol used in the security proof.
This is because the protocol for three-party QCKA with two-way classical communication depends only on the error rates associated with the stabilizer operators $\{\hat{S}_i\}$~\cite{2007Chen}, and these operators commute with the projectors onto the GHZ basis, $\{\ket{\Psi_{p,i_1,i_2}}\bra{\Psi_{p,i_1,i_2}}_{ABC}\}$.
Therefore, the measurement statistics—and hence the evaluated key rate—are completely determined by the GHZ-diagonal components, and the off-diagonal terms do not contribute to the security analysis.




Let $s_i$ be the error rate associated with the measurement of $\hat{S}_i$ for $i=1..7$, defined as the probability of obtaining the outcome -1. 
For $s_1$, $s_2$, $s_3$ and $s_4$, this definition is consistent with Eqs.~(\ref{defs1}) and (\ref{defs2s3}). 
Then, in the asymptotic limit, the error rates $\{s_i\}$ are given in terms of $\{p_{p,i_1,i_2}\}$ as follows:
\begin{equation}
    \label{ptos}
    \begin{split}
        &s_1 = p_{100} + p_{101} + p_{110} + p_{111},\\
        &s_2 = p_{010} + p_{011} + p_{110} + p_{111},\\
        &s_3 = p_{001} + p_{011} + p_{101} + p_{111},\\
        &s_4 = p_{010} + p_{110} + p_{001} + p_{101},\\
        &s_5 = p_{100} + p_{011} + p_{010} + p_{101},\\
        &s_6 = p_{100} + p_{011} + p_{110} + p_{001},\\
        &s_7 = p_{100} + p_{111} + p_{010} + p_{001}.
    \end{split}
\end{equation}
Alternatively, the probability amplitudes can be expressed in terms of the error rates,

\begin{equation}
    \label{ppij}
    \begin{split}
        p_{000} =& 1 - (s_1 + s_2 + s_3 + s_4 + s_5 + s_6 + s_7)/4,\\
        p_{100} =& (s_1 - s_2 - s_3 - s_4 + s_5 + s_6 + s_7)/4,\\
        p_{011} =& (- s_1 + s_2 + s_3 - s_4 + s_5 + s_6 - s_7)/4,\\
        p_{111} =& (s_1 + s_2 + s_3 - s_4 - s_5 - s_6 + s_7)/4,\\
        p_{010} =& (- s_1 + s_2 - s_3 + s_4 + s_5 - s_6 + s_7)/4,\\
        p_{110} =& (s_1 + s_2 - s_3 + s_4 - s_5 + s_6 - s_7)/4,\\
        p_{001} =& (- s_1 - s_2 + s_3 + s_4 - s_5 + s_6 + s_7)/4,\\
        p_{101} =& (s_1 - s_2 + s_3 + s_4 + s_5 - s_6 - s_7)/4.\\
    \end{split}
\end{equation}

Here, we review the B-step and its effect on the probability amplitudes and error rates.
The B-step procedure corresponds to a CNOT measurement by Alice, Bob and Charlie followed
by post selection. Specifically, Alice, Bob and Charlie will perform a CNOT gate between
a pair of their own qubits, measure the target qubit, and keep the control qubit only if the
target qubit measurements were the same for Alice, Bob and Charlie. 
Let $(q, j_1, j_2)$ denote the bit values labeling another GHZ-like state, similar to $(p, i_1, i_2)$.
In terms of this GHZ-basis notation, the two GHZ-like states transform as~\cite{2007Chen}  
\begin{equation}
    [(p, i_1, i_2), (q, j_1, j_2)] \rightarrow 
    [(p \oplus q, i_1, i_2), (q, i_1 \oplus j_1, i_2 \oplus j_2)],
\end{equation}
and the first GHZ state is kept only when $i_1 \oplus j_1 = i_2 \oplus j_2 = 0$.
Then, the original probability amplitudes transform as 
\begin{equation}
    \label{ppijtransform}
    \begin{pmatrix}
        p_{000} \\
        p_{100} \\
        p_{011} \\
        p_{111} \\
        p_{010} \\
        p_{110} \\
        p_{001} \\
        p_{101}
    \end{pmatrix}
\longrightarrow
    \begin{pmatrix}
        p_{000}^2 + p_{100}^2 \\
        2 p_{000} p_{100} \\
        p_{011}^2 + p_{111}^2 \\
        2 p_{011} p_{111} \\
        p_{010}^2 + p_{110}^2 \\
        2 p_{010} p_{110} \\
        p_{001}^2 + p_{101}^2 \\
        2 p_{001} p_{101}
    \end{pmatrix}
/p_{\rm pass},
\end{equation}
where $p_{\rm pass} 
= (p_{000} + p_{100})^{2}
+ (p_{001} + p_{101})^{2}
+ (p_{010} + p_{110})^{2}
+ (p_{011} + p_{111})^{2}$
is the `success' probability in which the first GHZ state is kept.
By applying the relations in Eq.~(\ref{ptos}) and (\ref{ppij}) to the updated $\{p_{p,i_1,i_2}\}$, the B-step updates the error rates $\bm{s}=(s_1,..,s_7)$ as follows:
\be
s_i \to  B_i(\bm{s}),
\label{sBtransform}
\ee
where we define
\begin{equation}
    \label{sB}
    \begin{split}
        B_1(\bm{x}) \coloneqq &
        \frac{
        -x_1^2 + x_2^2 + x_3^2 + x_4^2 - x_5^2  - x_6^2 - x_7^2
        + x_1 - x_2 - x_3 - x_4 + x_5 + x_6 + x_7
        }{
        2p_{\rm pass}
        },\\
        B_2(\bm{x}) \coloneqq &
        \frac{
        x_2^2 + x_3^2 + x_4^2 - 2x_3 x_4
        }{
        2p_{\rm pass}
        },\\
        B_3(\bm{x}) \coloneqq &
        \frac{
        x_2^2 + x_3^2 + x_4^2 - 2x_2 x_4
        }{
        2p_{\rm pass}
        },\\
        B_4(\bm{x}) \coloneqq &
        \frac{
        x_2^2 + x_3^2 + x_4^2 - 2x_2 x_3
        }{
        2p_{\rm pass}
        },\\
        B_5(\bm{x}) \coloneqq &
        \frac{
        x_2^2 + x_3^2 + x_4^2
        - 2x_1 x_5 - 2x_6 x_7
        + x_1 - x_2 - x_3 - x_4 + x_5 + x_6 + x_7
        }{
        2p_{\rm pass}
        },\\
        B_6(\bm{x}) \coloneqq &
        \frac{
        x_2^2 + x_3^2 + x_4^2
        - 2x_1 x_6 - 2x_5 x_7
        + x_1 - x_2 - x_3 - x_4 + x_5 + x_6 + x_7
        }{
        2p_{\rm pass}
        },\\
        B_7(\bm{x}) \coloneqq &
        \frac{
        x_2^2 + x_3^2 + x_4^2
        - 2x_1 x_7 - 2x_5 x_6
        + x_1 - x_2 - x_3 - x_4 + x_5 + x_6 + x_7
        }{
        2p_{\rm pass}
        },\\
    \end{split}
\end{equation}
with 
\be
p_{\rm pass} = 1 + x_2^2 + x_3^2 + x_4^2 - x_2 - x_3 - x_4.
\ee
Note that \(\bm{x}=(x_1,\ldots,x_7)\) denotes a vector of generic variables \(x_i\).

Next, we review the P-step and its effect on the error rates.
Let $(r, k_1, k_2)$ denote the bit values labeling another GHZ-like state, similar to $(p, i_1, i_2)$ and $(q, j_1, j_2)$.
The P-step can be regarded as implementing a three-qubit phase error correction code: if $p \oplus q = p\oplus r = 1 $, we apply $p \to p \oplus 1 $, otherwise keep the first GHZ-like state invariant. 
Under the P-step, the three GHZ states transform as~\cite{2007Chen} 
\be
[(p, i_1, i_2), (q, j_1, j_2), (r, k_1, k_2)]
 \to [(p, i_1 \oplus j_1 \oplus k_1, i_2 \oplus j_2 \oplus k_2), (p \oplus q, j_1, j_2), (p \oplus r, k_1, k_2)].
\ee
Then, the error rates $\bm{s}=(s_1,..,s_7)$ are updated under the P-step as follows:
\be
s_i \to P_i(\bm{s}),
\ee
where we define 
\begin{equation}
\label{sP}
\begin{split}
P_1(\bm{x}) &\coloneqq 3 x_1^2(1-x_1) + x_1^3, \\
P_2(\bm{x}) &\coloneqq 3x_2(1-x_2)^2 + x_2^3, \\
P_3(\bm{x}) &\coloneqq 3x_3(1-x_3)^2 + x_3^3, \\
P_4(\bm{x}) &\coloneqq 3x_4(1-x_4)^2 + x_4^3, \\
P_5(\bm{x}) &\coloneqq 6x_2^2x_5 -3x_2^2 - 6x_2x_5 + 3x_2 -2x_5^3+3x_5^2, \\
P_6(\bm{x}) &\coloneqq 6x_3^2x_6 -3x_3^2 - 6x_3x_6 + 3x_3 -2x_6^3+3x_6^2, \\
P_7(\bm{x}) &\coloneqq 6x_4^2x_7 -3x_4^2 - 6x_4x_7 + 3x_4 -2x_7^3+3x_7^2.
\end{split}
\end{equation}
Note that $P_1(\bm{x})$, $P_2(\bm{x})$, $P_3(\bm{x})$ and $P_4(\bm{x})$ have the same forms as in \cite{2003Gottesman}.

Based on the above, we formulate the main problem as follows. 
Let $\{\tilde{s}_i\}_{i=1..7}$ be error rates after a sequence of B-steps and P-steps. 
As in Eq.~(\ref{initialkeyrate}), the asymptotic secure key rate for the remaining rounds of the protocol is given by 
\begin{equation}
    \label{keyrateorig}
    1 - f\cdot\max\{h(\tilde{s}_2),h(\tilde{s}_3)\} -h(\tilde{s}_1).
\end{equation} 
In this work, only measurements in the $Z$ and $X$ bases are available, and therefore the error rates $s_5, s_6, s_7$, which require $Y$-basis measurements (see Eq.~(\ref{Stabelements})), remain undetermined. 
We therefore consider the worst-case scenario in which these parameters are controlled by Eve. 
From Eqs.~(\ref{sB}) and (\ref{sP}), the functions $B_2(\bm{s})$, $B_3(\bm{s})$, $B_4(\bm{s})$, $P_1(\bm{s})$, $P_2(\bm{s})$, $P_3(\bm{s})$ and $P_4(\bm{s})$ are independent of  $(s_5, s_6, s_7)$, whereas the remaining functions depend on these parameters. Consequently, $\tilde{s}_2$ and $\tilde{s}_3$ are independent of  $(s_5, s_6, s_7)$, while $\tilde{s}_1$ depends on them. 
Accordingly, Eq.~(\ref{keyrateorig}) shows that the main problem is to determine the physically realizable values of $(s_5, s_6, s_7)$ that maximize $h(\tilde{s}_1(s_5,s_6,s_7))$. 
The lower bound on the secure key rate is then given by
\be
\label{keyratepre}
R = 1 - f\cdot\max\{h(\tilde{s}_2),h(\tilde{s}_3)\} -\max\limits_{s_5,s_6,s_7} h(\tilde{s}_1(s_5,s_6,s_7)).
\ee

We discuss the constraints that $(s_5, s_6, s_7)$ must satisfy
in order to be physically realizable. Specifically, these constraints arise from the definition of
the probability amplitudes $p_{p,i_1,i_2}$, which must be non-negative and normalized.
From Eq.~(\ref{ppij}), the normalization condition is automatically satisfied, since the contributions from the $s_i$ cancel out when summing over all $p_{p,i_1,i_2}$. Then, the constraints on $s_5, s_6, s_7$ are solely determined by the non-negativity conditions $p_{p,i_1,i_2} \geq 0$, 
expressed in terms of the parameters $s_i$:

\begin{align}
    p_{011} \geq 0: \quad
    & s_5 + s_6 - s_7 \ge s_1 - s_2 - s_3 + s_4,\label{zero1}\\
    p_{111} \geq 0: \quad
    & s_5 + s_6 - s_7 \le s_1 + s_2 + s_3 - s_4,\label{zero2}\\
    p_{010} \geq 0: \quad
    & s_5 - s_6 + s_7 \ge s_1 - s_2 + s_3 - s_4,\label{zero3}\\
    p_{110} \geq 0: \quad
    & s_5 - s_6 + s_7 \le s_1 + s_2 - s_3 + s_4,\label{zero4}\\
    p_{001} \geq 0: \quad
    & -s_5 + s_6 + s_7 \ge s_1 + s_2 - s_3 - s_4,\label{zero5}\\
    p_{101} \geq 0: \quad
    & -s_5 + s_6 + s_7 \le s_1 - s_2 + s_3 + s_4,\label{zero6}\\
    p_{100} \geq 0: \quad
    & s_5 + s_6 + s_7 \ge -s_1 + s_2 + s_3 + s_4,\label{zero7}\\
    p_{000} \geq 0: \quad
    & s_5 + s_6 + s_7 \le 4 - (s_1 + s_2 + s_3 + s_4),\label{zero8} 
\end{align}
or in a slightly simpler form,
\begin{align}
    s_1 - (+s_2 + s_3 - s_4)
    &\le +s_5 + s_6 - s_7
    \le s_1 + (+s_2 + s_3 - s_4),\label{bound1}\\
    s_1 - (+s_2 - s_3 + s_4)
    &\le +s_5 - s_6 + s_7
    \le s_1 + (+s_2 - s_3 + s_4),\label{bound2}\\
    s_1 - (-s_2 + s_3 + s_4)
    &\le -s_5 + s_6 + s_7
    \le s_1 + (-s_2 + s_3 + s_4),\label{bound3}\\
    - s_1 + s_2 + s_3 + s_4
    &\le s_5 + s_6 + s_7
    \le 4 - (s_1 + s_2 + s_3 + s_4).\label{bound4}
\end{align}
Note that $p_{p, i_2, i_2} \leq 1$ because,
if $^\exists p_{p, i_1, i_2} > 1$,
the normalization condition cannot be satisfied for non-negative $p_{p, i_2, i_2}$.

\begin{figure}
 \centering
 \includegraphics[scale=0.3]
      {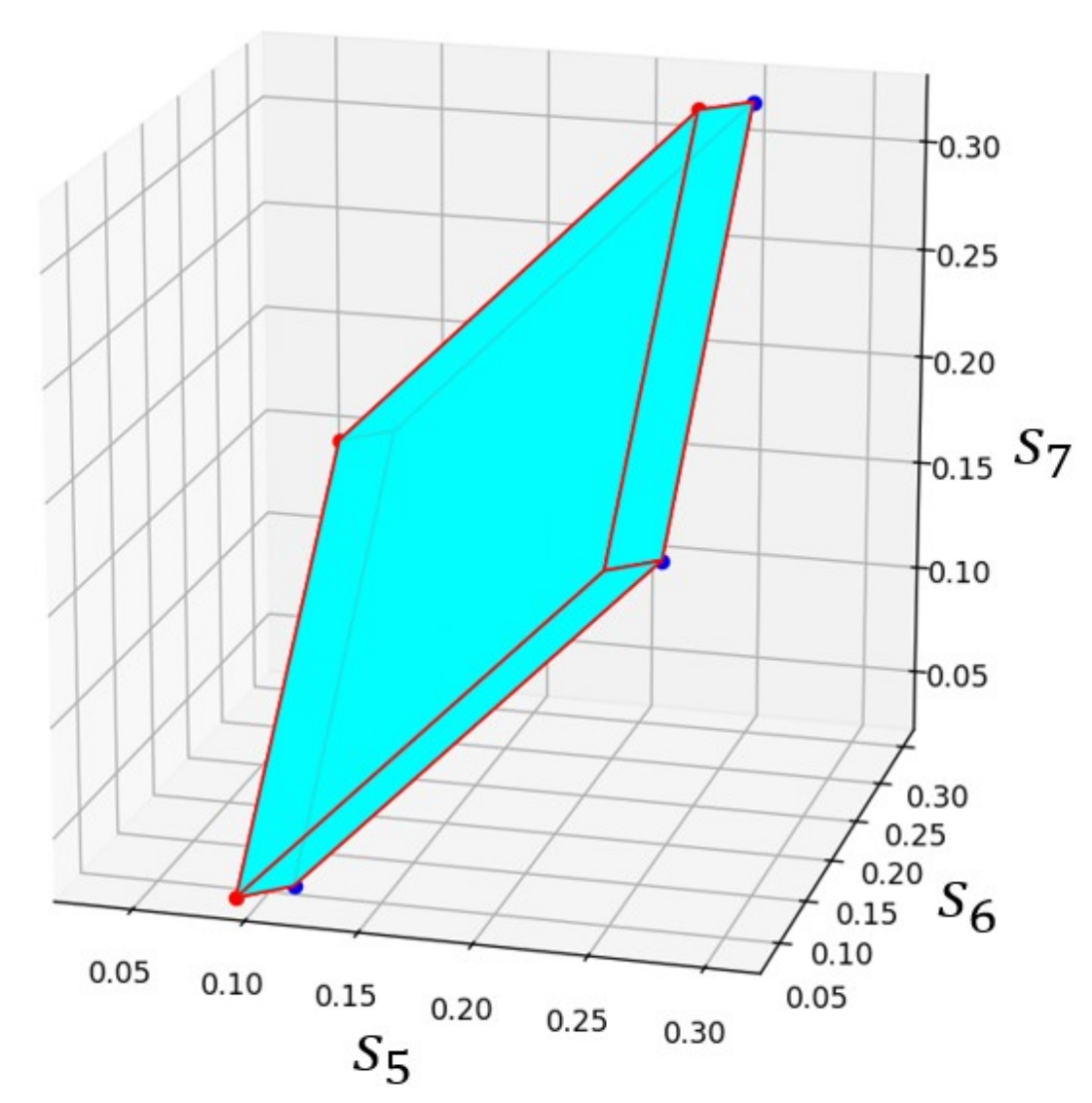}
 \caption{ An example of the physically realizable region \(T_{\rm phys}\) in the \((s_5,s_6,s_7)\) parameter space. Eqs.~(\ref{bound1})--(\ref{bound3}) define the interior of the parallelepiped. In this example, the boundary corresponding to Eq.~(\ref{bound4}) lies outside the parallelepiped.}
 \label{parallelopiped_fig}
\end{figure}


The bounds Eqs.~(\ref{bound1})--(\ref{bound4}) also imply conditions on $s_2, s_3$ and $s_4$.
That is, satisfying the following equations
are necessary conditions for any physically realizable $s_i$,
\begin{align}
    &0\le +s_2 + s_3 - s_4,\label{bound1_s1235}\\
    &0\le +s_2 - s_3 + s_4,\label{bound2_s1235}\\
    &0\le -s_2 + s_3 + s_4,\label{bound3_s1235}\\
    &s_2 + s_3 + s_4\le 2.\label{bound4_s1235}
\end{align}
Otherwise, there exists no $(s_5,s_6,s_7)$ that satisfies the physical constraints in Eqs.~(\ref{bound1})--(\ref{bound4}). 
In other words, Eqs.~(\ref{bound1_s1235})--(\ref{bound4_s1235}) are satisfied as long as  input states are physical.

Let $T_{\rm phys}$ denote the region of $(s_5,s_6,s_7)$ satisfying Eqs.~(\ref{bound1})--(\ref{bound4}). 
Eqs.~(\ref{bound1})--(\ref{bound3}) represent a parallelepiped in the $s_5, s_6, s_7$ 
parameter space as shown in Fig.~\ref{parallelopiped_fig}, with two additional bounds from Eq.~(\ref{bound4}). 
The optimization problem in Eq.~(\ref{keyratepre}) is then formulated as follows:
\be
\label{keyrate}
R = 1 - f\cdot\max\{h(\tilde{s}_2),h(\tilde{s}_3)\} -\max\limits_{(s_5,s_6,s_7) \in T_{\rm phys}} h(\tilde{s}_1(s_5,s_6,s_7)).
\ee

This optimization problem is similar to that considered for the BB84 or BBM92 protocols in  \cite{2003Gottesman}. 
The key difference is that, whereas only a single parameter is optimized in  \cite{2003Gottesman}, 
here three parameters must be optimized simultaneously. 
In the following section, we begin with the analysis of a single B-step, since $\tilde{s}_1$ is independent of $(s_5, s_6, s_7)$ when a single P-step is applied.

\section{Single B-step}
\label{Sec:singleB}

In this section, we derive an analytical solution to the optimization problem under the assumption that a single B-step is applied in the protocol. 
We present the objective function and the physical constraints, and solve the resulting convex optimization problem by exploiting its geometric structure.  
We show that there exists only eight distinct analytical solutions to the
optimization problem, depending on the initial error parameters, under the assumption of $s_1, s_2, s_3, s_4<1/2$ in Eq.~(\ref{assumptionlessthanhalf}). 

\subsection{Convex optimization formulation for secure key rate analysis}
For a single B-step, the secure key rate Eq.~(\ref{keyrate}) is 
\begin{equation}
    \label{keyrateB}
    R = 1 - f\cdot\max\{h(s_2^{(1)}),h(s_3^{(1)})\} - \max\limits_{(s_5,s_6,s_7) \in T_{\rm phys}} h(s_1^{(1)}(s_5,s_6,s_7)),
\end{equation}
where we define 
\be
s_1^{(1)}(s_5,s_6,s_7)= B_1(\bm{s}),~~s_2^{(1)}=B_2(\bm{s}),~~s_3^{(1)}=B_3(\bm{s}).
\label{def:s1s2s3forsingleB}
\ee  
Now we optimize the parameters $(s_5, s_6, s_7)$ to maximize $h(s_1^{(1)}(s_5,s_6,s_7))$. 
First, from the expressions in Eq.~(\ref{sB}), we observe that $s_1^{(1)}(s_5,s_6,s_7)$ is isotropic
with respect to $(s_5, s_6, s_7)$, and that its level surface forms a sphere in the $(s_5, s_6, s_7)$ parameter space. 
For $i=5,6,7$, the partial derivatives of $s_1^{(1)}(s_5,s_6,s_7)$ are 
\begin{equation}
    \label{partialdiff}
    \frac{\partial s_1^{(1)}(s_5,s_6,s_7)}{\partial s_i} = \frac{1-2s_i}{2p_{\rm pass}},
    ~~\frac{\partial^2 s_1^{(1)}(s_5,s_6,s_7)}{\partial s_i^2} = \frac{-2}{2p_{\rm pass}}.
\end{equation}
Since the denominator $2 p_{\rm pass}$ is
always positive, Eq.~(\ref{partialdiff}) indicates that $s_1^{(1)}(s_5,s_6,s_7)$ is maximized when
$s_5 = s_6 = s_7 = 1/2$.
Thus,
if we ignore the constraints in Eqs.~(\ref{bound1})--(\ref{bound4}),
the maximum value of $s_1^{(1)}(s_5,s_6,s_7)$ is obtained as

\begin{equation}
    \label{maxs1dash}
    s_1^{(1)}(1/2,1/2,1/2) = \frac{s_1(1 -s_1) + (\frac{3}{4} + s_2^2 + s_3^2 + s_4^2 - s_2 - s_3 - s_4)}
    {2 p_{\rm pass}}
    = \frac{s_1(1 -s_1) -\frac{1}{4}}{2p_{\rm pass}} + \frac{1}{2}
    . 
\end{equation}
Since $s_1<1/2$, we have $s_1(1-s_1)<1/4$, and therefore the maximum value is bounded as 
\be
s_1^{(1)}(1/2,1/2,1/2)<1/2.
\ee
Hence, 
\be
s_1^{(1)}(s_5,s_6,s_7) = B_1(\bm{s}) < 1/2
\label{B1supperbound}
\ee
holds for $s_1 < 1/2$. 
Since the binary entropy function $h(x)$ is monotonically
increasing on $[0,1/2]$,
maximizing $h(s_i^{(1)}(s_5,s_6,s_7))$ in Eq.~(\ref{keyrateB}) reduces to  
\be
\max\limits_{(s_5,s_6,s_7) \in T_{\rm phys}} s_1^{(1)}(s_5,s_6,s_7).
\label{maxs1problem}
\ee
Given the above observations, namely, \\
1. the spherical nature of the level surface of $s_i^{(1)}(s_5,s_6,s_7)$
with respect to $(s_5, s_6, s_7)$, \\
2. the fact that the worst is attained at $s_5 = s_6 = s_7 = 1/2$, \\ 
the solution to the optimization problem can be rephrased as follows:\\
\\
{\it The point $(s_5,s_6,s_7)$ that maximizes $s_i^{(1)}(s_5,s_6,s_7)$ is the point closest to $s_5 = s_6 = s_7 = 1/2$ within the region $T_{\rm phys}$.}\\
\\

In the next section, we derive analytical solutions of
the worst-case values of $(s_5, s_6, s_7)$ for the optimization problem in Eq. (\ref{maxs1problem}), i.e., the analytical solutions to the closest point to $s_5 = s_6 = s_7 = 1/2$ within the region $T_{\rm phys}$.

\subsection{Analytical solutions to the worst case $(s_5, s_6, s_7)$}
\label{subsection:analyticalsolutions}

In this section, we derive eight analytical solutions for the worst-case values of \((s_5,s_6,s_7)\), depending on the measurable parameters \(s_1,s_2,s_3,\) and \(s_4\).
We begin by analyzing the optimization problem in Eq.~(\ref{maxs1problem}) under the assumption that only the three constraints Eqs.~(\ref{zero2}), (\ref{zero4}), and (\ref{zero6}) are active among Eqs.~(\ref{zero1})--(\ref{zero8}):
\begin{equation}
\max_{(s_5,s_6,s_7)\in T_{\rm pyramid}}
s_1^{(1)}(s_5,s_6,s_7),
\label{problemTpyramid}
\end{equation}
where \(T_{\rm pyramid}\) denotes the region in the \((s_5,s_6,s_7)\) space satisfying Eqs.~(\ref{zero2}), (\ref{zero4}), and (\ref{zero6}). 
By its definition, the region $T_{\rm pyramid}$ includes the original physically realizable region $T_{\rm phys}$, namely,
\be
T_{\rm phys} \subseteq T_{\rm pyramid}.
\ee
One representative analytical solution to Eq.~(\ref{problemTpyramid}) is derived explicitly in the main text, while the remaining seven solutions are presented in Appendix~\ref{Appe:othercases}. 
We then show that the obtained analytical solutions automatically satisfy the remaining constraints, Eqs.~(\ref{zero1}), (\ref{zero3}), (\ref{zero5}), (\ref{zero7}), and (\ref{zero8}), regardless of the values of \(s_1,s_2,s_3,\) and \(s_4\). The verification of these remaining constraints is provided in Appendix~\ref{Appe:constraints}. As a consequence, these solutions are also solutions to the optimization problem in Eq.~(\ref{maxs1problem}).

Geometrically, the three constraints in Eqs.~(\ref{zero2}), (\ref{zero4}), and (\ref{zero6}) correspond to three faces spanning the $(s_5, s_6, s_7)$
parameter space, where its position depends on the measurable parameters $s_1, s_2, s_3$ and $s_4$.
As the faces are not parallel, the region $T_{\rm pyramid}$ corresponds to  a triangular pyramid 
in the $(s_5, s_6, s_7)$ parameter space as shown in Fig.~\ref{pyramid2_fig}. The analytical solutions to the worst-case 
$(s_5, s_6, s_7)$ is then equivalent to the point within $T_{\rm pyramid}$ which is
geometrically closest to $s_5 = s_6 = s_7 = 1/2$.

The eight cases can be divided into four situations. That is, the one case where the worst point lies on the corner of $T_{\rm pyramid}$, the three cases where the worst point lies on one of the three edges, the three cases where the worst point lies on one of the three faces of $T_{\rm pyramid}$, 
 and the one case where the worst point lies
within $T_{\rm pyramid}$.
In order to express the eight solutions and their conditions in detail, we define relevant vectors.
First,
we use the parameter space
in the coordinate system of $(s_5, s_6, s_7)$. 
The faces of $T_{\rm pyramid}$ are determined by the equality conditions corresponding to Eqs.~(\ref{zero2}), (\ref{zero4}), and (\ref{zero6}), namely,
\begin{align}
    & s_5 + s_6 - s_7 = s_1 + s_2 + s_3 - s_4,\label{surface_e3}\\
    & s_5 - s_6 + s_7 = s_1 + s_2 - s_3 + s_4,\label{surface_e2}\\
    & -s_5 + s_6 + s_7 = s_1 - s_2 + s_3 + s_4.\label{surface_e1}
\end{align}
The coordinates of the corner are obtained by simultaneously solving Eqs.~(\ref{surface_e3})--(\ref{surface_e1}). The corner of the pyramid is therefore given by the point \(\bm{S}\),
\begin{equation}
    \bm{S} =
    \begin{pmatrix}
        s_1 + s_2\\
        s_1 + s_3\\
        s_1 + s_4
    \end{pmatrix}.
\end{equation}
Additionally, we define the point $\bm{P} = (1/2, 1/2, 1/2)^\mathrm{T}$, and denote $\bm{P}'$ as the worst-case point of $(s_5, s_6, s_7)$, which is equivalent to the point closest to $\bm{P}$ within the
region defined by the pyramid. These points are shown in Fig.~\ref{pyramid2_fig} (a).

\begin{figure}
 \centering
 \includegraphics[scale=0.7]
    {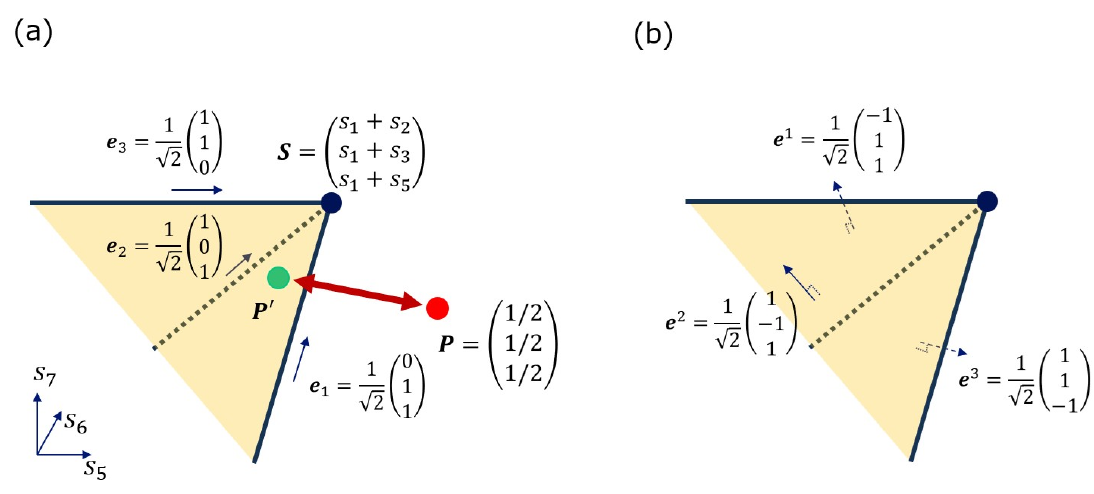}
    \caption{The triangular pyramid $T_{\rm pyramid}$ defined in the $s_5, s_6, s_7$ parameter space by the constraints
    Eq.~(\ref{zero2}), Eq.~(\ref{zero4}), Eq.~(\ref{zero6}). 
    The corner of the pyramid corresponds to the upper-right corner of the parallelepiped in Fig.~\ref{parallelopiped_fig}. 
    (a) The edge unit vectors (covariant vectors) $\{\bm{e}_i\}$ , the peak point $\bm{S}$, the worst point $\bm{P}$ and the worst point $\bm{P}'$ within $T_{\rm pyramid}$. (b) The contravariant vectors $\{\bm{e}^i\}$.
}
 \label{pyramid2_fig}
\end{figure}

The signs of $s_5, s_6, s_7$ in Eqs.~(\ref{surface_e3})--(\ref{surface_e1}) give the normal vectors of the corresponding
faces, therefore the outward normal vectors from the interior of $T_{\rm pyramid}$ without normalization are
$\bm{n}^1 = (-1, 1, 1)^{\mathrm{T}},\bm{n}^2 = (1, -1, 1)^{\mathrm{T}},\bm{n}^3 = (1, 1, -1)^{\mathrm{T}}$
for Eq.~(\ref{surface_e1}), Eq.~(\ref{surface_e2}), Eq.~(\ref{surface_e3}), respectively.
The unit vectors of the edges can be obtained by taking the cross product of the corresponding face normals.
For example, for the edge corresponding to the
intersection of Eqs.~(\ref{surface_e3}) and (\ref{surface_e2}) 
(i.e., the edge corresponding to $\bm{e_1}$ in Fig.~\ref{pyramid2_fig} (a)), the edge vector should be  proportional to a vector perpendicular to 
the normal vectors of Eqs.~(\ref{surface_e3}) and (\ref{surface_e2}),
\begin{equation}
    \begin{pmatrix}
        1\\
        -1\\
        1\\
    \end{pmatrix}
    \times
    \begin{pmatrix}
        1\\
        1\\
        -1\\
    \end{pmatrix}
    =
    \begin{pmatrix}
        0\\
        2\\
        2\\
    \end{pmatrix},
\end{equation}
and likewise for the other two edges.
Then, the edge unit vectors pointing toward the corner $\bm{S}$ can be defined as 
\begin{align}
    \bm{e}_1 =
    \frac{1}{\sqrt{2}}
    \begin{pmatrix}
        0\\
        1\\
        1\\
    \end{pmatrix},~~
    \bm{e}_2 =
    \frac{1}{\sqrt{2}}
    \begin{pmatrix}
        1\\
        0\\
        1\\
    \end{pmatrix},~~
    \bm{e}_3 =
    \frac{1}{\sqrt{2}}
    \begin{pmatrix}
        1\\
        1\\
        0\\
    \end{pmatrix},\label{edgeunitvector3}
\end{align}
 such that $\bm{e}_i \cdot \bm{n}^i >0$ and $|\bm{e}_i| = 1$.
These are shown in Fig.~\ref{pyramid2_fig} (a). 
Finally, we will define the contravariant vectors, $\bm{e}^i$, which take the form,
\begin{align}
    \bm{e}^1 =
    \frac{1}{\sqrt{2}}
    \begin{pmatrix}
        -1\\
        1\\
        1\\
    \end{pmatrix},~~
    \bm{e}^2 =
    \frac{1}{\sqrt{2}}
    \begin{pmatrix}
        1\\
        -1\\
        1\\
    \end{pmatrix},~~
    \bm{e}^3 =
    \frac{1}{\sqrt{2}}
    \begin{pmatrix}
        1\\
        1\\
        -1\\
    \end{pmatrix},
\end{align} 
shown in Fig.\ref{pyramid2_fig} (b). The contravariant vectors
are parallel to the normal vectors of the corresponding faces,
and the covariant (edge) and contravariant vectors satisfy the following relation,

\begin{equation}
     \bm{e}^i \cdot \bm{e}_j = \bm{e}_i \cdot \bm{e}^j = \delta_{ij},
\end{equation}
 where $\delta_{ij}$ is the Kronecker delta.
It should be noted that, in order to satisfy the above relation, the contravariant
vector is not a unit vector where $|\bm{e}^i|^2 = \frac{3}{2}$. 
The covariant and contravariant vectors also satisfy the following relations,
\begin{align}
    \bm{e}_i \cdot \bm{e}_j &= \frac{1}{2},\\
    \bm{e}^i \cdot \bm{e}^j &= -\frac{1}{2}.
\end{align}

In the following analysis, the origin of the coordinates is taken as $\bm{S}$.
Then, the point $\bm{P}$ can be expressed in terms of the covariant and contravariant vectors as
\begin{equation}
\begin{split}
    \bm{SP} &= \left(\frac{1}{2} - s_1 -s_2, \frac{1}{2} - s_1 -s_3, \frac{1}{2} - s_1 -s_4\right)^\mathrm{T}\\
    &= v^1\bm{e}_1 + v^2\bm{e}_2 + v^3\bm{e}_3\\
    &= v_1\bm{e}^1 + v_2\bm{e}^2 + v_3\bm{e}^3,\\
\end{split}
\end{equation}
where $v^i, v_j$ are coefficients. The expressions of $v^i, v_j$ can be obtained by
multiplication of the appropriate vector and $\bm{SP}$ as,

\begin{equation}
\label{e^1}
\begin{split}
    \bm{e}^1 \cdot \bm{SP} &= v^1\bm{e}^1\cdot\bm{e}_1 + v^2\bm{e}^1\cdot\bm{e}_2 + v^3\bm{e}^1\cdot\bm{e}_3\\
    &= v^1 (\because \bm{e}^i\cdot\bm{e}_j=\delta_{ij})\\
    &= \frac{1}{\sqrt{2}} (-1,1,1) \cdot \left(\frac{1}{2} - s_1 -s_2, \frac{1}{2} - s_1 -s_3, \frac{1}{2} - s_1 -s_4\right)^\mathrm{T}\\
    &= \frac{1}{\sqrt{2}} \left(\frac{1}{2} -s_1 +s_2 -s_3 -s_4\right),
\end{split}
\end{equation}

\begin{equation}
\label{e_1}
\begin{split}
    \bm{e}_1 \cdot \bm{SP} &= v_1\bm{e}_1\cdot\bm{e}^1 + v_2\bm{e}_1\cdot\bm{e}^2 + v_3\bm{e}_1\cdot\bm{e}^3\\
    &= v_1 (\because \bm{e}_i\cdot\bm{e}^j=\delta_{ij})\\
    &= \frac{1}{\sqrt{2}} (0, 1, 1) \cdot  \left(\frac{1}{2} - s_1 -s_2, \frac{1}{2} - s_1 -s_3, \frac{1}{2} - s_1 -s_4\right)^\mathrm{T}\\
    &= \frac{1}{\sqrt{2}} (1 -2s_1 -s_3 -s_4).\\
\end{split}
\end{equation}
The other $v^i, v_j$ can be obtained using similar methods as
\begin{align}
    v^2 &= \frac{1}{\sqrt{2}} \left(\frac{1}{2} -s_1 -s_2 +s_3 -s_4\right),\label{e^2}\\
    v^3 &= \frac{1}{\sqrt{2}} \left(\frac{1}{2} -s_1 -s_2 -s_3 +s_4\right),\label{e^3}\\
    v_2 &= \frac{1}{\sqrt{2}} (1 -2s_1 -s_2 -s_4),\label{e_2}\\
    v_3 &= \frac{1}{\sqrt{2}} (1 -2s_1 -s_2 -s_3).\label{e_3}
\end{align}
Using these coordinates, the conditions and solutions for the 
eight separate cases discussed above can be derived.
 In the following, we present the analytical solution only for the corner case, since only this case yields a positive key rate, as numerically confirmed in Sec.~\ref{Sec:singleBstepNumerical}. The remaining seven solutions corresponding to edge and face cases are shown in Appendix~\ref{Appe:othercases}.  
The results for all cases are summarized in Table~\ref{Table:solutions}.

\begin{table}[htbp]
\centering
\begin{tabular}{Sc|Sc|Sc}
\hline
Case & Conditions for $s_1,s_2,s_3,s_4$ & Solutions $(s_5,s_6,s_7)$ \\
\hline

1
&
$\displaystyle
\max\{(s_2+s_3),(s_3+s_4),(s_2+s_4)\} \le 1-2s_1
$
&
$(s_1+s_2,\, s_1+s_3,\, s_1+s_4)$
\\

\hline

2
&
$\displaystyle
\begin{aligned}
1 &< 2s_1+s_3+s_4,\\
|s_3-s_4| &\le 1-2s_1-2s_2
\end{aligned}
$
&
$\displaystyle
\left(
s_1+s_2,
\frac{1}{2}(s_3-s_4+1),
\frac{1}{2}(-s_3+s_4+1)
\right)
$
\\

\hline

3
&
$\displaystyle
\begin{aligned}
1 &< 2s_1+s_2+s_4,\\
|s_2-s_4| &\le 1-2s_1-2s_3
\end{aligned}
$
&
$\displaystyle
\left(
\frac{1}{2}(s_2-s_4+1),
s_1+s_3,
\frac{1}{2}(-s_2+s_4+1)
\right)
$
\\

\hline

4
&
$\displaystyle
\begin{aligned}
1 &< 2s_1+s_2+s_3,\\
|s_2-s_3| &\le 1-2s_1-2s_4
\end{aligned}
$
&
$\displaystyle
\left(
\frac{1}{2}(s_2-s_3+1),
\frac{1}{2}(-s_2+s_3+1),
s_1+s_4
\right)
$
\\

\hline

5
&
$\displaystyle
\begin{aligned}
s_1+(-s_2+s_3+s_4) &\le \frac{1}{2},\\
1-2s_1-s_2 &< \min\{(2s_3-s_4),(-s_3+2s_4)\}
\end{aligned}
$
&
$\displaystyle
\left(
\frac{1}{3}(2-s_1+s_2-s_3-s_4),
\frac{1}{3}(1+s_1-s_2+s_3+s_4),
\frac{1}{3}(1+s_1-s_2+s_3+s_4)
\right)
$
\\

\hline

6
&
$\displaystyle
\begin{aligned}
s_1+(s_2-s_3+s_4) &\le \frac{1}{2},\\
1-2s_1-s_3 &< \min\{(-s_2+2s_4),(2s_2-s_4)\}
\end{aligned}
$
&
$\displaystyle
\left(
\frac{1}{3}(1+s_1+s_2-s_3+s_4),
\frac{1}{3}(2-s_1-s_2+s_3-s_4),
\frac{1}{3}(1+s_1+s_2-s_3+s_4)
\right)
$
\\

\hline

7
&
$\displaystyle
\begin{aligned}
s_1+(s_2+s_3-s_4) &\le \frac{1}{2},\\
1-2s_1-s_4 &< \min\{(-s_2+2s_3),(2s_2-s_3)\}
\end{aligned}
$
&
$\displaystyle
\left(
\frac{1}{3}(1+s_1+s_2+s_3-s_4),
\frac{1}{3}(1+s_1+s_2+s_3-s_4),
\frac{1}{3}(2-s_1-s_2-s_3+s_4)
\right)
$
\\

\hline

8
&
$\displaystyle
\frac{1}{2}-s_1 < \min\{(-s_2+s_3+s_4), (s_2-s_3+s_4), (s_2+s_3-s_4)\}
$
&
$(1/2,1/2,1/2)$
\\

\hline
\end{tabular}
\caption{Conditions for $s_1,s_2,s_3,s_4$ and corresponding solutions $(s_5,s_6,s_7)$ to the optimization problem in Eq.~(\ref{problemTpyramid}). Case~1 corresponds to the corner solution discussed in the main text. Case~2--4 correspond to edge solutions, Case~5--7 to face solutions, and Case~8 to an interior solution of the pyramid. The analytical derivations for Case~2--8 are presented in Appendix~\ref{Appe:othercases}.}
\label{Table:solutions}
\end{table}

Now we consider the case in which the worst-case point lies at the corner \(\bm{S}\), defined by the intersection of the three faces. Geometrically, this corresponds to the situation where the point \(\bm{P}\) is closer to \(\bm{S}\) than to any other candidate solution on the edges or faces of the pyramid $T_{\rm pyramid}$. 
The corner $\bm{S}$ is the closest point in $T_{\rm pyramid}$ to $\bm{P}$ if and only if the displacement vector $\bm{SP}$ belongs to the normal cone of $T_{\rm pyramid}$ at $\bm{S}$. Since this normal cone is generated by the outward normals of the three faces, the condition is
\begin{equation}
\bm{SP} \in \mathrm{cone}(\bm{e}^1,\bm{e}^2,\bm{e}^3),
\end{equation}
namely,
\begin{equation}
\exists \nu_i \ge 0,\ i=1,2,3,\ \text{s.t.}\
\bm{SP}
=
\nu_1 \bm{e}^1
+
\nu_2 \bm{e}^2
+
\nu_3 \bm{e}^3.
\end{equation}
The analytical solution of the worst point $\bm{P}'$ is simply given by $\bm{S}$:
\begin{equation}
    \label{solutioncorner}
    \bm{P'}= \bm{S}=
     \begin{pmatrix}
        s_1+s_2\\
        s_1+s_3\\
        s_1+s_4\\
    \end{pmatrix},
\end{equation}
if and only if $v_1,v_2,v_3\geq0$ is satisfied, 
or in terms of $s_i$,
\begin{align}
    1-2s_1-s_3-s_4&\geq0,\label{cornercond1}\\
    1-2s_1-s_2-s_4&\geq0,\label{cornercond2}\\
    1-2s_1-s_2-s_3&\geq0,\label{cornercond3}
\end{align}
namely,
\be
\label{cornercondsum}
\max\{(s_2+s_3),(s_3+s_4),(s_2+s_4)\} \le 1-2s_1. 
\ee

As shown in Appendix~\ref{satisfyingotherconstraintscorner}, the solution in Eq.~(\ref{solutioncorner}) satisfies the constraint in Eqs.~(\ref{zero1}), (\ref{zero3}), (\ref{zero5}), (\ref{zero7}) and (\ref{zero8}). Therefore, this is also a solution for the original optimization problem in Eq.~(\ref{maxs1problem}) if Eq.~(\ref{cornercondsum}) is satisfied.

\subsection{Numerical analysis}
\label{Sec:singleBstepNumerical}
Here, we briefly present a representative result of our numerical analysis to motivate the discussion in the following section. A more detailed analysis, including a comparison with previous work, is provided in Sec.~\ref{Sec:Numerical}. The secure key rate $R$ in Eq.~(\ref{keyrateB}) is obtained by substituting the solutions $(s_5,s_6,s_7)$ listed in Table~{\ref{Table:solutions}} into the expression for $s_1^{(1)}(s_5,s_6,s_7)$ in Eq.~(\ref{def:s1s2s3forsingleB}). Therefore, $R$ is determined by $s_1$, $s_2$, $s_3$ and $s_4$. 
Setting $f=1$, we evaluate $R$ over a large number of $(s_1,s_2,s_3,s_4)$ configurations. Using a grid search, we generate 450,704 parameter sets that satisfy the physical constraints in Eqs.~(\ref{bound1_s1235})--(\ref{bound4_s1235}). Among these, 141,074 correspond to Case 1 (corner) in Table~\ref{Table:solutions}, of which 37,122 yield a positive key rate $(R>0)$. In contrast, among the remaining 309,630 parameter sets corresponding to cases other than Case 1, none yield a positive key rate. These results indicate that, in practice, only Case 1 is relevant at least  when one B-step is performed in the protocol. Motivated by this observation, we restrict our analysis to Case 1 in the following section when considering multiple B-steps.

\section{Multiple B-steps}
\label{Sec:multiB}

In this section, we derive an analytical solution to the optimization problem in Eq.~(\ref{keyrate}) when an arbitrary number of B-steps are applied in the protocol.
First, for $i=1,\dots,7$, we define
\begin{equation}
s_i^{(0)}(s_5,s_6,s_7)=s_i.
\label{def:si0}
\end{equation}
Then, for $n\geq1$, we define the error rates after applying the B-step $n$ times as
\begin{equation}
s_i^{(n)}(s_5,s_6,s_7)
=
B_i(\bm{s}^{(n-1)}(s_5,s_6,s_7)),
\label{def:sin}
\end{equation}
where
\begin{equation}
\bm{s}^{(n-1)}(s_5,s_6,s_7)
=
(
s_1^{(n-1)}(s_5,s_6,s_7),
\dots,
s_7^{(n-1)}(s_5,s_6,s_7)
).
\end{equation}
For $i=2,3,4$, the quantities $s_i^{(n)}$ are independent of
$(s_5,s_6,s_7)$, and thus we simply denote them by
$s_i^{(n)}$.
From Eq.~(\ref{B1supperbound}), if $s_1<1/2$, then
\begin{equation}
B_1(\bm{s})<\frac{1}{2}.
\end{equation}
By recursively applying this property, we obtain
\begin{equation}
\label{s1nrange}
s_1^{(n)}(s_5,s_6,s_7)
<
\frac{1}{2}
\end{equation}
for all $n\geq0$.
Since the binary entropy function $h(x)$ is monotonically increasing for
$0\leq x<1/2$, maximizing
$h(s_1^{(n)}(s_5,s_6,s_7))$
is equivalent to maximizing
$s_1^{(n)}(s_5,s_6,s_7)$ itself.
Accordingly, for $n\geq0$, we define
\begin{equation}
\label{def:s1nmax}
s_{1,\mathrm{max}}^{(n)}
\coloneqq
\max_{(s_5,s_6,s_7) \in T_{\rm phys}}
s_1^{(n)}(s_5,s_6,s_7).
\end{equation}
The secure key rate in Eq.~(\ref{keyrate}) is therefore written as
\begin{equation}
\label{SKRformultiinitial}
R
=
1
-
f\cdot
\max\{
h(\tilde{s}_2),
h(\tilde{s}_3)
\}
-
h(s_{1,\mathrm{max}}^{(n)}).
\end{equation}

From the formulations of $\{B_i(\bm{x})\}$ in Eq.~(\ref{sB}), the objective function $s_1^{(n)}(s_5,s_6,s_7)$ becomes increasingly complicated as the number of B-step iterations increases. However, we show that the optimization problem after multiple B-step iterations can be reduced to a recursive sequence of single-step optimizations, thereby avoiding the need to directly optimize this increasingly complicated function. In particular, when ($s_1,s_2,s_3,s_4$) for $n=0$ belong to Case 1 (corner) defined in Sec.~\ref{Sec:singleB}, 
namely,
\be
{\rm max} \{(s_2 +s_3), (s_3+s_4), (s_2+s_4)\} < 1- 2 s_1,
\ee
the optimal solution for any number of B-steps can be shown to remain at
\be
(s_5,s_6,s_7) = (s_1+s_2, s_1+s_3, s_1+s_4),
\ee
which is identical to the optimal solution for a single B-step.
This result is established later in Theorem~\ref{theo:main}. 
To prove the theorem, we first derive three auxiliary lemmas. 
Lemma~\ref{lemma:bound} provides an upper bound on $s_{1, \rm max}^{(n+1)}$ in terms of $s_{1, \rm max}^{(n)}$, which enables the recursive analysis of the worst-case phase error. 
Lemma~\ref{lemma:case} shows that updated parameters are in Case 1 if they are in Case 1 for $n=0$. 
Lemma~\ref{lemma:optimality} shows that the optimal corner solution remains optimal after the parameter update.

First, the following lemma establishes a recursive upper bound relating the worst-case phase error after $n$ and $n+1$ B-steps.
\begin{lemma}
\label{lemma:bound}
Suppose that $s_{1,\mathrm{max}}^{(n)}$, $s_2^{(n)}$, $s_3^{(n)}$, and $s_4^{(n)}$ satisfy
\begin{equation}
\max \{s_2^{(n)}+s_3^{(n)},\, s_3^{(n)}+s_4^{(n)},\, s_2^{(n)}+s_4^{(n)}\}
< 1-2s_{1,\mathrm{max}}^{(n)}.
\label{conditioncornernstep}
\end{equation}
Then,
\begin{equation}
\begin{split}
s_{1,\mathrm{max}}^{(n+1)}
=
\max_{(s_5,s_6,s_7) \in T_{\rm phys}}
s_1^{(n+1)}(s_5,s_6,s_7)
\leq
B_1(&s_{1,\mathrm{max}}^{(n)},
s_2^{(n)},
s_3^{(n)},
s_4^{(n)}, \\
&
s_{1,\mathrm{max}}^{(n)}+s_2^{(n)},
s_{1,\mathrm{max}}^{(n)}+s_3^{(n)},
s_{1,\mathrm{max}}^{(n)}+s_4^{(n)}).
\end{split}
\label{eq:bsbound}
\end{equation}
The equality holds if there exists $(s_5,s_6,s_7)$ satisfying
\begin{align}
s_1^{(n)}(s_5,s_6,s_7)
&=
s_{1,\mathrm{max}}^{(n)}, \\
s_5^{(n)}(s_5,s_6,s_7)
&=
s_{1,\mathrm{max}}^{(n)}+s_2^{(n)}, \\
s_6^{(n)}(s_5,s_6,s_7)
&=
s_{1,\mathrm{max}}^{(n)}+s_3^{(n)}, \\
s_7^{(n)}(s_5,s_6,s_7)
&=
s_{1,\mathrm{max}}^{(n)}+s_4^{(n)}.
\end{align}
\end{lemma}

\begin{proof}
Let $p_{p,i_1,i_2}^{(n)}~(p,i_1,i_2\in\{0,1\})$ denote the probability distribution obtained after applying the B-step transformation Eq.~(\ref{ppijtransform}) $n$ times.
Even after $n$ B-steps, the relations between $\{p_{p,i_1,i_2}^{(n)}\}$ and $\{s_i^{(n)}\}$ given in Eqs.~(\ref{ptos}) and (\ref{ppij}) remain valid.
In particular,
\begin{equation}
\label{pkpk}
\begin{split}
p_{101}^{(n)}
=&
\frac{
s_1^{(n)}
-s_2^{(n)}
+s_3^{(n)}
+s_4^{(n)}
+s_5^{(n)}
-s_6^{(n)}
-s_7^{(n)}
}{4},\\
p_{110}^{(n)}
=&
\frac{
s_1^{(n)}
+s_2^{(n)}
-s_3^{(n)}
+s_4^{(n)}
-s_5^{(n)}
+s_6^{(n)}
-s_7^{(n)}
}{4},\\
p_{111}^{(n)}
=&
\frac{
s_1^{(n)}
+s_2^{(n)}
+s_3^{(n)}
-s_4^{(n)}
-s_5^{(n)}
-s_6^{(n)}
+s_7^{(n)}
}{4},
\end{split}
\end{equation}
where we abbreviated
$s_i^{(n)} \equiv s_i^{(n)}(s_5,s_6,s_7)$ for $i=1,5,6,7$.
From the physical constraints
$p_{101}^{(n)}\geq0$,
$p_{110}^{(n)}\geq0$,
and
$p_{111}^{(n)}\geq0$,
we obtain
\begin{equation}
\begin{split}
-s_5^{(n)}+s_6^{(n)}+s_7^{(n)}
&\leq
s_1^{(n)}-s_2^{(n)}+s_3^{(n)}+s_4^{(n)},\\
+s_5^{(n)}-s_6^{(n)}+s_7^{(n)}
&\leq
s_1^{(n)}+s_2^{(n)}-s_3^{(n)}+s_4^{(n)},\\
+s_5^{(n)}+s_6^{(n)}-s_7^{(n)}
&\leq
s_1^{(n)}+s_2^{(n)}+s_3^{(n)}-s_4^{(n)}.
\end{split}
\end{equation}
Since
$s_1^{(n)}\leq s_{1,\mathrm{max}}^{(n)}$,
it follows that
\begin{equation}
\label{upper}
\begin{split}
-s_5^{(n)}+s_6^{(n)}+s_7^{(n)}
&\leq
s_{1,\mathrm{max}}^{(n)}
-s_2^{(n)}
+s_3^{(n)}
+s_4^{(n)},\\
+s_5^{(n)}-s_6^{(n)}+s_7^{(n)}
&\leq
s_{1,\mathrm{max}}^{(n)}
+s_2^{(n)}
-s_3^{(n)}
+s_4^{(n)},\\
+s_5^{(n)}+s_6^{(n)}-s_7^{(n)}
&\leq
s_{1,\mathrm{max}}^{(n)}
+s_2^{(n)}
+s_3^{(n)}
-s_4^{(n)}.
\end{split}
\end{equation}
Introducing free parameters $(s_5',s_6',s_7')$ independent of $(s_5,s_6,s_7)$, define the region
\begin{equation}
\label{Tpyramid}
\begin{split}
T_{\mathrm{pyramid}}^{(n)}
\coloneqq
\{
(s_5',s_6',s_7')
~|~
&
-s_5'+s_6'+s_7'
\leq
s_{1,\mathrm{max}}^{(n)}
-s_2^{(n)}
+s_3^{(n)}
+s_4^{(n)},\\
&
+s_5'-s_6'+s_7'
\leq
s_{1,\mathrm{max}}^{(n)}
+s_2^{(n)}
-s_3^{(n)}
+s_4^{(n)},\\
&
+s_5'+s_6'-s_7'
\leq
s_{1,\mathrm{max}}^{(n)}
+s_2^{(n)}
+s_3^{(n)}
-s_4^{(n)}
\}.
\end{split}
\end{equation}
Equations~(\ref{upper}) and (\ref{Tpyramid}) imply that the accessible region of
$(s_5^{(n)},s_6^{(n)},s_7^{(n)})$
is contained in
$T_{\mathrm{pyramid}}^{(n)}$.
Therefore,
\begin{equation}
\begin{split}
s_{1,\mathrm{max}}^{(n+1)}
&=
\max_{(s_5,s_6,s_7)\in T_{\rm phys}}
B_1(
s_1^{(n)},
s_2^{(n)},
s_3^{(n)},
s_4^{(n)},
s_5^{(n)},
s_6^{(n)},
s_7^{(n)}
) \\
&\leq
\max_{(s_5',s_6',s_7')\in T_{\mathrm{pyramid}}^{(n)}}
\max_{(s_5,s_6,s_7)\in T_{\rm phys}}
B_1(
s_1^{(n)},
s_2^{(n)},
s_3^{(n)},
s_4^{(n)},
s_5',
s_6',
s_7'
).
\end{split}
\end{equation}
Since
\begin{equation}
\frac{\partial B_1(\bm{x})}{\partial x_1}
=
\frac{1-2x_1}{p_{\mathrm{pass}}},
\end{equation}
$B_1(\bm{x})$ is monotonically increasing in $x_1$ for $x_1<1/2$.
Using Eq.~(\ref{s1nrange}),
$s_1^{(n)}<1/2$ holds, and hence
\begin{equation}
\begin{split}
\max_{(s_5,s_6,s_7)\in T_{\rm phys}}
B_1(
s_1^{(n)},
s_2^{(n)},
s_3^{(n)},
s_4^{(n)},
s_5',
s_6',
s_7'
)
=
B_1(
s_{1,\mathrm{max}}^{(n)},
s_2^{(n)},
s_3^{(n)},
s_4^{(n)},
s_5',
s_6',
s_7'
).
\end{split}
\end{equation}
Thus,
\begin{equation}
s_{1,\mathrm{max}}^{(n+1)}
\leq
\max_{(s_5',s_6',s_7')\in T_{\mathrm{pyramid}}^{(n)}}
B_1(
s_{1,\mathrm{max}}^{(n)},
s_2^{(n)},
s_3^{(n)},
s_4^{(n)},
s_5',
s_6',
s_7'
).
\end{equation}
The right-hand side is exactly the same optimization problem as Eq.~(\ref{problemTpyramid}) for a single B-step.
Condition~(\ref{conditioncornernstep}) corresponds to the corner condition in Eq.~(\ref{cornercondsum}).
Therefore, from Eq.~(\ref{solutioncorner}), the maximum is attained at
\begin{equation}
s_5'
=
s_{1,\mathrm{max}}^{(n)}+s_2^{(n)},
\quad
s_6'
=
s_{1,\mathrm{max}}^{(n)}+s_3^{(n)},
\quad
s_7'
=
s_{1,\mathrm{max}}^{(n)}+s_4^{(n)}.
\end{equation}
This proves Eq.~(\ref{eq:bsbound}).
\end{proof}

Next, the following lemma shows that once the parameters belong to the corner regime, they remain in the same regime under arbitrary iterations of the B-step.

\begin{lemma}
\label{lemma:case}
If
\begin{equation}
\max \{s_2+s_3,\, s_3+s_4,\, s_2+s_4\}
<
1-2s_1,
\end{equation}
then, for any $n\geq0$,
\begin{equation}
\label{eq:theooptimality}
\max \{
s_2^{(n)}+s_3^{(n)},
s_3^{(n)}+s_4^{(n)},
s_2^{(n)}+s_4^{(n)}
\}
<
1-2s_{1,\mathrm{max}}^{(n)}
\end{equation}
holds.
\end{lemma}

\begin{proof}
For $n=0$, Eq.~(\ref{eq:theooptimality}) follows immediately from
$s_i^{(0)}=s_i$.
Assume that, for $n=k$,
\begin{equation}
\label{eq:theooptimalityfork}
\max \{
s_2^{(k)}+s_3^{(k)},
s_3^{(k)}+s_4^{(k)},
s_2^{(k)}+s_4^{(k)}
\}
<
1-2s_{1,\mathrm{max}}^{(k)}
\end{equation}
holds.
For simplicity, define
\begin{equation}
t_1 \coloneqq s_{1,\mathrm{max}}^{(k)},
\quad
t_2 \coloneqq s_2^{(k)},
\quad
t_3 \coloneqq s_3^{(k)},
\quad
t_4 \coloneqq s_4^{(k)}.
\end{equation}
Then Eq.~(\ref{eq:theooptimalityfork}) can be rewritten as
\begin{equation}
\label{eq:theooptimalityforkwithx}
2t_1+t_2+t_3<1,
\quad
2t_1+t_3+t_4<1,
\quad
2t_1+t_2+t_4<1.
\end{equation}
From Lemma~\ref{lemma:bound} and Eq.~(\ref{sB}),
\begin{equation}
\begin{split}
s_{1,\mathrm{max}}^{(k+1)}
&\leq
B_1(
t_1,
t_2,
t_3,
t_4,
t_1+t_2,
t_1+t_3,
t_1+t_4
) \\
&=
\frac{
2t_1\left(
1-t_1-\frac{1}{2}(t_2+t_3+t_4)
\right)
}{
p_{\mathrm{pass}}
},
\end{split}
\end{equation}
where
\begin{equation}
p_{\mathrm{pass}}
=
1-t_2-t_3-t_4+t_2^2+t_3^2+t_4^2.
\end{equation}
On the other hand, from the expressions of
$B_2$, $B_3$, and $B_4$ in Eq.~(\ref{sB}),
\begin{equation}
\begin{split}
s_2^{(k+1)}
&=
\frac{
t_2^2+(t_3-t_4)^2
}{
2p_{\mathrm{pass}}
},\\
s_3^{(k+1)}
&=
\frac{
t_3^2+(t_2-t_4)^2
}{
2p_{\mathrm{pass}}
},\\
s_4^{(k+1)}
&=
\frac{
t_4^2+(t_2-t_3)^2
}{
2p_{\mathrm{pass}}
}.
\end{split}
\end{equation}
Therefore,
\begin{equation}
2s_{1,\mathrm{max}}^{(k+1)}
+
s_2^{(k+1)}
+
s_3^{(k+1)}
\leq
\frac{
N_{23}
}{
p_{\mathrm{pass}}
},
\end{equation}
where
\begin{equation}
N_{23}
=
4t_1
-4t_1^2
-2t_1(t_2+t_3+t_4)
+t_2^2+t_3^2+t_4^2
-t_4(t_2+t_3).
\end{equation}
A direct calculation gives
\begin{equation}
\begin{split}
p_{\mathrm{pass}}-N_{23}
&=
1
-4t_1
-t_2
-t_3
-t_4
+4t_1^2
+2t_1(t_2+t_3+t_4)
+t_4(t_2+t_3) \\
&=
(1-2t_1-t_2-t_3)(1-2t_1-t_4).
\end{split}
\end{equation}
Using Eq.~(\ref{eq:theooptimalityforkwithx}),
both factors on the right-hand side are positive, and hence
$p_{\mathrm{pass}}-N_{23}>0$.
Since $p_{\mathrm{pass}}\geq0$, it follows that
$N_{23}/p_{\mathrm{pass}}<1$.
Therefore,
\begin{equation}
\label{conditionk1for23}
2s_{1,\mathrm{max}}^{(k+1)}
+
s_2^{(k+1)}
+
s_3^{(k+1)}
<
1.
\end{equation}

By exchanging $t_4$ and $t_2$ in the above argument, we similarly obtain
\begin{equation}
\label{conditionk1for34}
2s_{1,\mathrm{max}}^{(k+1)}
+
s_3^{(k+1)}
+
s_4^{(k+1)}
<
1,
\end{equation}
and by exchanging $t_4$ and $t_3$,
\begin{equation}
\label{conditionk1for24}
2s_{1,\mathrm{max}}^{(k+1)}
+
s_2^{(k+1)}
+
s_4^{(k+1)}
<
1.
\end{equation}
Equations~(\ref{conditionk1for23})--(\ref{conditionk1for24}) imply
\begin{equation}
\max \{
s_2^{(k+1)}+s_3^{(k+1)},
s_3^{(k+1)}+s_4^{(k+1)},
s_2^{(k+1)}+s_4^{(k+1)}
\}
<
1-2s_{1,\mathrm{max}}^{(k+1)}.
\end{equation}
Thus, by mathematical induction,
Eq.~(\ref{eq:theooptimality}) holds for all $n\geq0$.
\end{proof}

Next, define
\begin{equation}
\begin{split}
s_{1,c}^{(0)}
&=
s_1,\\
s_{5,c}^{(0)}
&=
s_1+s_2,\\
s_{6,c}^{(0)}
&=
s_1+s_3,\\
s_{7,c}^{(0)}
&=
s_1+s_4,
\end{split}
\label{defsicfor0}
\end{equation}
and for $n\geq0$ and $i=1,5,6,7$,
\begin{equation}
\label{sic}
s_{i,c}^{(n)}
=
s_i^{(n)}(s_1+s_2,s_1+s_3,s_1+s_4).
\end{equation}
Namely, $s_{i,c}^{(n)}$ denotes the error rates after applying the B-step $n$ times with the initial condition $(s_5,s_6,s_7) = (s_1+s_2,s_1+s_3,s_1+s_4)$, which corresponds to the corner solution.
From Eq.~(\ref{def:sin}), for $n\geq1$,
\begin{equation}
\label{ralationsicBstep}
\begin{split}
s_{i,c}^{(n)}
=
B_i(
s_{1,c}^{(n-1)},
s_2^{(n-1)},
s_3^{(n-1)},
s_4^{(n-1)},
s_{5,c}^{(n-1)},
s_{6,c}^{(n-1)},
s_{7,c}^{(n-1)}
).
\end{split}
\end{equation}
We then obtain the following lemma, which shows that the corner structure is preserved under repeated B-steps.

\begin{lemma}
\label{lemma:optimality}
For all $n\geq0$,
\begin{equation}
\label{eq:optimality}
s_{5,c}^{(n)}
=
s_{1,c}^{(n)}+s_2^{(n)},
\quad
s_{6,c}^{(n)}
=
s_{1,c}^{(n)}+s_3^{(n)},
\quad
s_{7,c}^{(n)}
=
s_{1,c}^{(n)}+s_4^{(n)}.
\end{equation}
\end{lemma}

\begin{proof}
For $n=0$, Eq.~(\ref{eq:optimality}) follows directly from Eq.~(\ref{defsicfor0}).
Assume that for $n=k$,
\begin{equation}
\label{kassumptioncorner}
s_{5,c}^{(k)}
=
s_{1,c}^{(k)}+s_2^{(k)},
\quad
s_{6,c}^{(k)}
=
s_{1,c}^{(k)}+s_3^{(k)},
\quad
s_{7,c}^{(k)}
=
s_{1,c}^{(k)}+s_4^{(k)}
\end{equation}
holds.
By setting
\begin{equation}
(s_5,s_6,s_7)
=
(s_1+s_2,s_1+s_3,s_1+s_4),
\end{equation}
Eq.~(\ref{sic}) implies
\begin{equation}
s_i^{(k)}(s_5,s_6,s_7)
=
s_{i,c}^{(k)}.
\end{equation}
Substituting this into Eq.~(\ref{pkpk}), we obtain
\begin{equation}
\begin{split}
p_{101}^{(k)}
=&
\frac{
s_{1,c}^{(k)}
-s_2^{(k)}
+s_3^{(k)}
+s_4^{(k)}
+s_{5,c}^{(k)}
-s_{6,c}^{(k)}
-s_{7,c}^{(k)}
}{4},\\
p_{110}^{(k)}
=&
\frac{
s_{1,c}^{(k)}
+s_2^{(k)}
-s_3^{(k)}
+s_4^{(k)}
-s_{5,c}^{(k)}
+s_{6,c}^{(k)}
-s_{7,c}^{(k)}
}{4},\\
p_{111}^{(k)}
=&
\frac{
s_{1,c}^{(k)}
+s_2^{(k)}
+s_3^{(k)}
-s_4^{(k)}
-s_{5,c}^{(k)}
-s_{6,c}^{(k)}
+s_{7,c}^{(k)}
}{4}.
\end{split}
\end{equation}
Using Eq.~(\ref{kassumptioncorner}), we find
\begin{equation}
p_{101}^{(k)}
=
p_{110}^{(k)}
=
p_{111}^{(k)}
=
0.
\end{equation}
From the transition rule Eq.~(\ref{ppijtransform}),
it follows that
\begin{equation}
p_{101}^{(k+1)}
=
p_{110}^{(k+1)}
=
p_{111}^{(k+1)}
=
0.
\end{equation}
Substituting these conditions into Eq.~(\ref{pkpk}) at step $k+1$, we obtain
\begin{equation}
s_{5,c}^{(k+1)}
=
s_{1,c}^{(k+1)}+s_2^{(k+1)},
\quad
s_{6,c}^{(k+1)}
=
s_{1,c}^{(k+1)}+s_3^{(k+1)},
\quad
s_{7,c}^{(k+1)}
=
s_{1,c}^{(k+1)}+s_4^{(k+1)}.
\end{equation}
Thus, by mathematical induction, Eq.~(\ref{eq:optimality}) holds for all $n\geq0$.
\end{proof}

Combining the above lemmas, we obtain the following main theorem.

\begin{thm}
\label{theo:main}
Suppose that $s_1$, $s_2$, $s_3$, and $s_4$ satisfy
\begin{equation}
\label{theorem1eq}
\max \{s_2+s_3,\, s_3+s_4,\, s_2+s_4\}
<
1-2s_1.
\end{equation}
Then, the worst-case phase error rate after applying the B-step $n~(\geq0)$ times, defined in Eq.~(\ref{def:s1nmax}), is given by
\begin{equation}
s_{1,\mathrm{max}}^{(n)}
=
s_{1,c}^{(n)}.
\label{eq:maintheo}
\end{equation}
\end{thm}

\begin{proof}
For $n=0$, Eq.~(\ref{eq:maintheo}) follows immediately from the definitions
Eqs.~(\ref{def:si0}) and (\ref{defsicfor0}).
Assume that
\begin{equation}
s_{1,c}^{(k)}
=
s_{1,\mathrm{max}}^{(k)}
\end{equation}
holds for some $k\geq0$.
From Lemma~\ref{lemma:optimality},
\begin{equation}
\begin{split}
s_{5,c}^{(k)}
&=
s_{1,\mathrm{max}}^{(k)}+s_2^{(k)},\\
s_{6,c}^{(k)}
&=
s_{1,\mathrm{max}}^{(k)}+s_3^{(k)},\\
s_{7,c}^{(k)}
&=
s_{1,\mathrm{max}}^{(k)}+s_4^{(k)}.
\end{split}
\end{equation}
Recalling that
\begin{equation}
s_{i,c}^{(k)}
=
s_i^{(k)}(s_1+s_2,s_1+s_3,s_1+s_4)
\end{equation}
in Eq.~(\ref{sic}), there exists
\begin{equation}
(s_5,s_6,s_7)
=
(s_1+s_2,s_1+s_3,s_1+s_4)
\end{equation}
such that
\begin{equation}
\begin{split}
s_1^{(k)}(s_5,s_6,s_7)
&=
s_{1,\mathrm{max}}^{(k)},\\
s_5^{(k)}(s_5,s_6,s_7)
&=
s_{1,\mathrm{max}}^{(k)}+s_2^{(k)},\\
s_6^{(k)}(s_5,s_6,s_7)
&=
s_{1,\mathrm{max}}^{(k)}+s_3^{(k)},\\
s_7^{(k)}(s_5,s_6,s_7)
&=
s_{1,\mathrm{max}}^{(k)}+s_4^{(k)}.
\end{split}
\end{equation}
Furthermore, by Lemma~\ref{lemma:case},
\begin{equation}
\max \{
s_2^{(k)}+s_3^{(k)},
s_3^{(k)}+s_4^{(k)},
s_2^{(k)}+s_4^{(k)}
\}
<
1-2s_{1,\mathrm{max}}^{(k)}.
\end{equation}
Therefore, the equality condition in Lemma~\ref{lemma:bound} is satisfied, yielding
\begin{equation}
\begin{split}
s_{1,\mathrm{max}}^{(k+1)}
&=
B_1(
s_{1,\mathrm{max}}^{(k)},
s_2^{(k)},
s_3^{(k)},
s_4^{(k)},
s_{1,\mathrm{max}}^{(k)}+s_2^{(k)},
s_{1,\mathrm{max}}^{(k)}+s_3^{(k)},
s_{1,\mathrm{max}}^{(k)}+s_4^{(k)}
)\\
&=
B_1(
s_{1,c}^{(k)},
s_2^{(k)},
s_3^{(k)},
s_4^{(k)},
s_{5,c}^{(k)},
s_{6,c}^{(k)},
s_{7,c}^{(k)}
)\\
&=
s_{1,c}^{(k+1)},
\end{split}
\end{equation}
where the last equality follows from Eq.~(\ref{ralationsicBstep}).
Hence, by mathematical induction,
\begin{equation}
s_{1,\mathrm{max}}^{(n)}
=
s_{1,c}^{(n)}
\end{equation}
holds for all $n\geq0$.
\end{proof}

Finally, we summarize the secure key rate and its recursive evaluation procedure.
From Eq.~(\ref{SKRformultiinitial}) and Theorem~\ref{theo:main}, the secure key rate $R$ is given by
\begin{equation}
R
=
1
-
\max\{h(s_2^{(n)}),\, h(s_3^{(n)})\}
-
h(s_{1,c}^{(n)}).
\label{keyratemultiBstep}
\end{equation}
From Eqs.~(\ref{ralationsicBstep}) and (\ref{eq:optimality}), for $k\geq0$,
\begin{equation}
\begin{split}
s_{1,c}^{(k+1)}
&=
B_1(
s_{1,c}^{(k)},
s_2^{(k)},
s_3^{(k)},
s_4^{(k)},
s_{1,c}^{(k)}+s_2^{(k)},
s_{1,c}^{(k)}+s_3^{(k)},
s_{1,c}^{(k)}+s_4^{(k)}
)\\
&=
\frac{
2s_{1,c}^{(k)}
\left(
1-s_{1,c}^{(k)}
-\frac{1}{2}
(s_2^{(k)}+s_3^{(k)}+s_4^{(k)})
\right)
}{
p_{\mathrm{pass}}^{(k)}
},
\end{split}
\end{equation}
where
\begin{equation}
p_{\mathrm{pass}}^{(k)}
=
1
-
s_2^{(k)}
-
s_3^{(k)}
-
s_4^{(k)}
+
(s_2^{(k)})^2
+
(s_3^{(k)})^2
+
(s_4^{(k)})^2.
\end{equation}
Meanwhile,
\begin{equation}
\begin{split}
s_2^{(k+1)}
&=
\frac{
(s_2^{(k)})^2
+
(s_3^{(k)}-s_4^{(k)})^2
}{
2p_{\mathrm{pass}}^{(k)}
},\\
s_3^{(k+1)}
&=
\frac{
(s_3^{(k)})^2
+
(s_2^{(k)}-s_4^{(k)})^2
}{
2p_{\mathrm{pass}}^{(k)}
},\\
s_4^{(k+1)}
&=
\frac{
(s_4^{(k)})^2
+
(s_2^{(k)}-s_3^{(k)})^2
}{
2p_{\mathrm{pass}}^{(k)}
}.
\end{split}
\end{equation}
By recursively applying these relations for
$k=0,1,\dots,n-1$,
one obtains
$s_{1,c}^{(n)}$,
$s_2^{(n)}$,
and
$s_3^{(n)}$,
which completely determine the secure key rate in Eq.~(\ref{keyratemultiBstep}).

\section{P-step and its combination to B-step}
\label{Sec:multiP}
In this section, we analyze the P-step and its combination with the B-step.
For a single P-step, Eq.~(\ref{sP}) gives
\begin{equation}
P_1(s_1)
=
3s_1^2(1-s_1)+s_1^3
=
3s_1^2-2s_1^3.
\end{equation}
Unlike the B-step case, the phase error after a P-step does not depend on
$s_5$, $s_6$, or $s_7$.
Therefore, no optimization problem needs to be considered.
The same observation applies even when multiple P-steps are performed.

For the combination of B-steps and P-steps, we consider only the case in which P-steps are applied after the B-steps (e.g., BBBPPP).
This is because the results derived in the previous section can be directly applied only to this ordering.
At least for BB84 and BBM92, it has been reported that this ordering yields a higher error threshold than the opposite ordering (e.g., PPPBBB) or alternating applications (e.g., BPBPBP)~\cite{2003Gottesman}.

For $i=1,\dots,7$, define
\begin{equation}
s_i^{(n,0)}(s_5,s_6,s_7)
=
s_i^{(n)}(s_5,s_6,s_7),
\end{equation}
where $s_i^{(n)}(s_5,s_6,s_7)$ denotes the error rate after applying the B-step $n$ times, as defined in Eq.~(\ref{def:sin}). 
Then, for $m\geq1$, we define the error rates after applying the B-step $n$ times followed by the P-step $m$ times as
\begin{equation}
s_i^{(n,m)}(s_5,s_6,s_7)
=
P_i(\bm{s}^{(n,m-1)}(s_5,s_6,s_7)).
\end{equation}
Here,
\begin{equation}
\bm{s}^{(n,m-1)}(s_5,s_6,s_7)
=
(
s_1^{(n,m-1)}(s_5,s_6,s_7),
\dots,
s_7^{(n,m-1)}(s_5,s_6,s_7)
).
\end{equation}
For $i=2,3,4$,
$s_i^{(n,m)}(s_5,s_6,s_7)$
is independent of
$(s_5,s_6,s_7)$,
and we therefore simply denote it by
$s_i^{(n,m)}$.
Since
\begin{equation}
\frac{dP_1(x_1)}{dx_1}
=
6x_1-6x_1^2
=
6x_1(1-x_1)
\geq0
\qquad
(0\leq x_1\leq1),
\label{monoincreP}
\end{equation}
$P_1(x_1)$ is a monotonically increasing function of $x_1$.
Furthermore,
\begin{equation}
P_1\left(\frac{1}{2}\right)
=
\frac{1}{2},
\end{equation}
and therefore
$x_1<1/2$
implies
$P_1(x_1)<1/2$.
From Eq.~(\ref{s1nrange}), the phase error satisfies
$s_1^{(n)}(s_5,s_6,s_7)<1/2$
when only B-steps are applied.
Hence, even after subsequent P-steps,
\begin{equation}
\label{eq:Prange}
s_1^{(n,m)}(s_5,s_6,s_7)
<
\frac{1}{2}
\end{equation}
holds.
Therefore, maximizing
$h(s_1^{(n,m)}(s_5,s_6,s_7))$
is equivalent to maximizing
$s_1^{(n,m)}(s_5,s_6,s_7)$ itself.

We now obtain the following theorem.
Namely, even after applying P-steps following the B-steps, the worst-case solution remains identical to that of the single-step B-step optimization.
\begin{thm}
If
\begin{equation}
\max \{
s_2+s_3,\,
s_3+s_4,\,
s_2+s_4
\}
<
1-2s_1,
\end{equation}
then, for all $n\geq0$ and $m\geq0$,
\begin{equation}
\max_{(s_5,s_6,s_7)\in T_{\rm phys}}
s_1^{(n,m)}(s_5,s_6,s_7)
=
s_1^{(n,m)}(s_1+s_2,s_1+s_3,s_1+s_4)
\label{eq:Psteptheo}
\end{equation}
holds.
\end{thm}

\begin{proof}
For $m=0$, Eq.~(\ref{eq:Psteptheo}) follows directly from Theorem~\ref{theo:main}.
Assume that, for some $k\geq0$,
\begin{equation}
\max_{(s_5,s_6,s_7)\in T_{\rm phys}}
s_1^{(n,k)}(s_5,s_6,s_7)
=
s_1^{(n,k)}(s_1+s_2,s_1+s_3,s_1+s_4)
\end{equation}
holds.
Then,
\begin{equation}
\begin{split}
\max_{(s_5,s_6,s_7)\in T_{\rm phys}}
s_1^{(n,k+1)}(s_5,s_6,s_7)
&=
\max_{s_5,s_6,s_7}
P_1(s_1^{(n,k)}(s_5,s_6,s_7)) \\
&=
P_1\left(
\max_{(s_5,s_6,s_7)\in T_{\rm phys}}
s_1^{(n,k)}(s_5,s_6,s_7)
\right) \\
&=
P_1(
s_1^{(n,k)}(s_1+s_2,s_1+s_3,s_1+s_4)
) \\
&=
s_1^{(n,k+1)}(s_1+s_2,s_1+s_3,s_1+s_4).
\end{split}
\end{equation}
In the second equality, we used the fact that
$P_1(x_1)$
is monotonically increasing in $x_1$, as shown in Eq.~(\ref{monoincreP}).
Therefore, by mathematical induction,
Eq.~(\ref{eq:Psteptheo}) holds for all $m\geq0$.
\end{proof}

Summarizing all the above results, when
\begin{equation}
\max \{
s_2+s_3,\,
s_3+s_4,\,
s_2+s_4
\}
<
1-2s_1,
\end{equation}
the secure key rate after applying the B-step $n~(\geq0)$ times followed by the P-step $m~(\geq0)$ times is given by
\begin{equation}
R
=
1
-
\max\{
h(s_2^{(n,m)}),
h(s_3^{(n,m)})
\}
-
h(s_1^{(n,m)}(s_1+s_2,s_1+s_3,s_1+s_4))
.
\end{equation}
The quantities
$s_1^{(n,m)}(s_1+s_2,s_1+s_3,s_1+s_4)$,
$s_2^{(n,m)}$,
and
$s_3^{(n,m)}$
are obtained by applying
$P_1(x_1)$,
$P_2(x_2)$,
and
$P_3(x_3)$,
respectively, $m$ times to
$s_{1,c}^{(n)}$,
$s_2^{(n)}$,
and
$s_3^{(n)}$
obtained after the $n$ B-steps.

\section{Numerical analysis}
\label{Sec:Numerical}
In this section, we numerically investigate how the error threshold for non-negative secure key rate \(R\) is improved by repeated applications of the B-step and the P-step. 
Setting the error-correction efficiency to \(f=1\), we consider the symmetric parameter setting
\[
s_1=s_X, \qquad s_2=s_3=s_4=s_Z,
\]
and plot the threshold curve \((s_Z,s_X)\) satisfying \(R=0\) in Fig.~\ref{Threshold_fig}. 
The solid curves correspond to the threshold obtained from our secure key rate formula \(R\) when the B-step is applied \(0,1,2,3,4,\) and \(5\) times, respectively, ordered from lower left to upper right. The colors are blue, orange, green, red, purple, and brown, respectively. The region below each curve corresponds to the parameter region where secure key extraction is possible. All of these regions satisfy the condition of Theorem~1, Eq.~(\ref{theorem1eq}), namely,
\[
s_Z+s_X<\frac12,
\]
under the assumption \(s_2=s_3=s_4=s_Z\). The line \(s_Z+s_X=1/2\) is shown as a dotted line in the figure.
The dashed curves represent the thresholds obtained in the previous work analyzing a single B-step~\cite{2026Krawec}. Compared with the case of a single B-step considered in the previous work, the threshold is significantly improved by applying two or more B-steps. In particular, focusing on the symmetric case \(s_Z=s_X\), the threshold is approximately \(11\%\) without B-step and approximately \(15\%\) with one B-step, whereas it exceeds \(20\%\) when five B-steps are applied.

\begin{figure}
 \centering
 \includegraphics[scale=0.35]
    {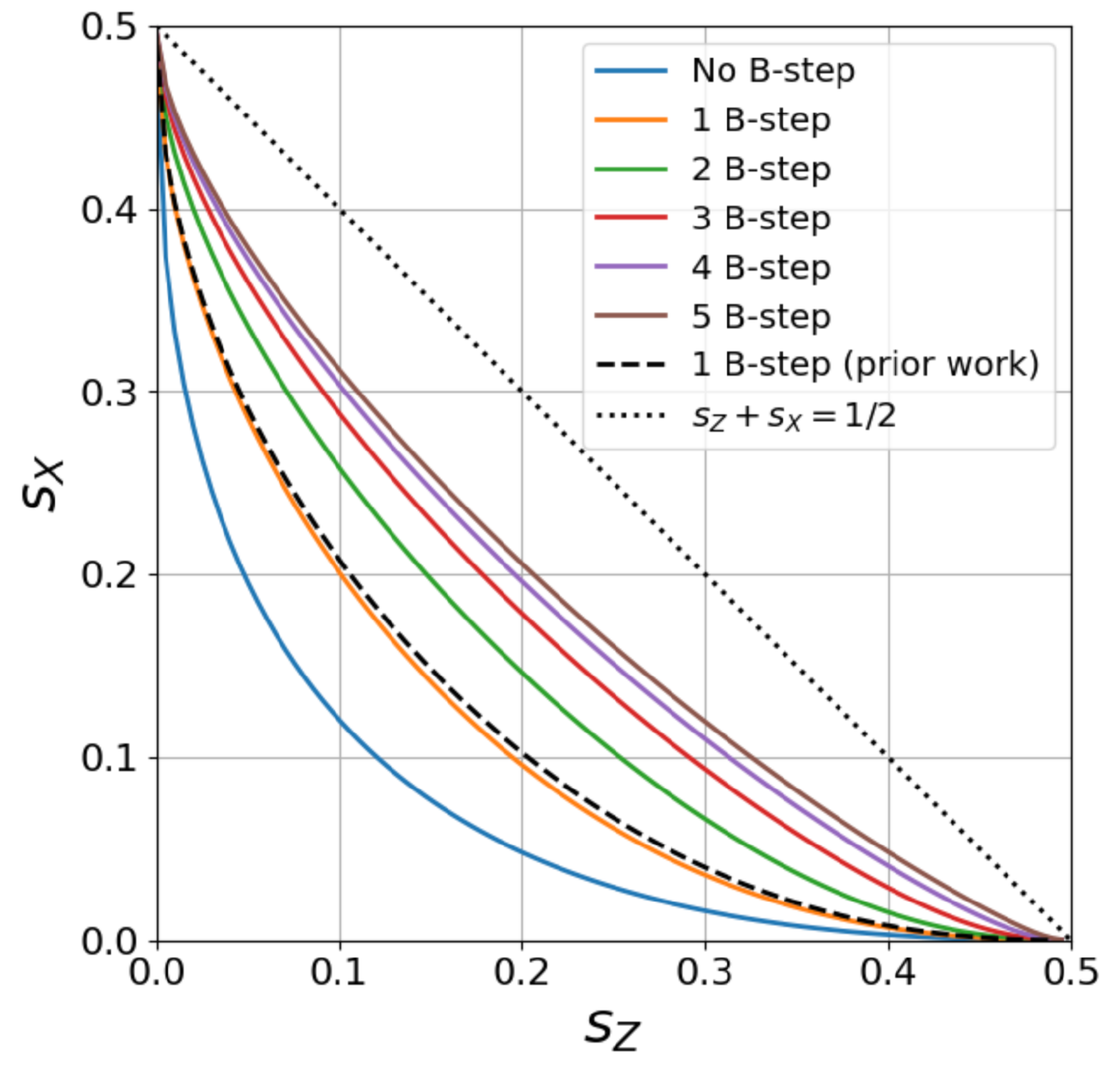}
    \caption{
Threshold curves $(s_Z,s_X)$ for non-negative secure key rate $R$ under the symmetric setting $s_1=s_X$ and $s_2=s_3=s_4=s_Z$, with error-correction efficiency $f=1$. The solid curves correspond to $0,1,2,3,4,$ and $5$ applications of the B-step (from lower left to upper right). The region below each curve yields a non-negative secure key rate. The dashed curve shows the threshold obtained by using the previous analysis of a single B-step~\cite{2026Krawec}, and the dotted line indicates $s_Z+s_X=1/2$, corresponding to the condition in Theorem~1. Repeated applications of the B-step significantly improve the threshold; in particular, along the symmetric line $s_Z=s_X$, the tolerable error rate increases from approximately $11\%$ without a B-step and $15\%$ with a single B-step to over $20\%$ after five B-steps.
}
 \label{Threshold_fig}
\end{figure}

As far as our numerical analysis is concerned, we did not observe any improvement of the threshold by applying a P-step after the B-step, in contrast to the result in \cite{2003Gottesman}. Qualitatively, this is because the increase in \(s_Z\) caused by the P-step dominates the reduction in \(s_X\). More specifically, the behavior can be understood as follows.
As the B-step is repeatedly applied, \(s_X\) approaches \(1/2\), while \(s_Z\) approaches \(0\). We therefore assume
\begin{equation}
s_X=\frac12-\epsilon_X,
\qquad
s_Z=\epsilon_Z,
\qquad
(\epsilon_X,\epsilon_Z\ll1).
\end{equation}
The secure key rate is given by
\begin{equation}
R=1-h(s_X)-h(s_Z).
\label{keyratePstepdisc}
\end{equation}
Suppose that a single P-step transforms the error rates as
\[
s_X\to s'_X,
\qquad
s_Z\to s'_Z.
\]
Then,
\begin{equation}
s'_X=P_1(s_X)
=
\frac12-\frac32\epsilon_X+2\epsilon_X^3
=
\frac12-\frac32\epsilon_X+O(\epsilon_X^3).
\end{equation}
Using the Taylor expansion of the binary entropy around \(1/2\),
\begin{equation}
h\!\left(\frac12-\epsilon\right)
=
1-\frac{2\epsilon^2}{\ln2}
+O(\epsilon^4),
\end{equation}
we obtain
\begin{equation}
1-h(s_X)
=
\frac{2\epsilon_X^2}{\ln2}
+O(\epsilon_X^4),
\qquad
1-h(s'_X)
=
\frac{2}{\ln2}
\left(\frac32\epsilon_X\right)^2
+O(\epsilon_X^4),
\label{epsilonsX}
\end{equation}
and therefore
\begin{equation}
1-h(s'_X)
=
\frac94\bigl(1-h(s_X)\bigr)
+O(\epsilon_X^4).
\end{equation}
On the other hand,
\begin{equation}
s'_Z=P_2(s_Z)
=
3\epsilon_Z(1-\epsilon_Z)^2+\epsilon_Z^3
=
3\epsilon_Z+O(\epsilon_Z^2),
\end{equation}
and
\begin{equation}
h(\epsilon)
=
-\epsilon\log_2\epsilon
-(1-\epsilon)\log_2(1-\epsilon)
=
-\epsilon\log_2\epsilon
+O(\epsilon).
\end{equation}
Hence,
\begin{equation}
\label{hsZ}
h(s_Z)
=
-\epsilon_Z\log_2\epsilon_Z
+O(\epsilon_Z),
\end{equation}
and
\begin{equation}
h(s'_Z)
=
-3\epsilon_Z\log_23\epsilon_Z
+O(\epsilon_Z)
=
-3\epsilon_Z\log_2\epsilon_Z
+O(\epsilon_Z).
\label{hsdashZ}
\end{equation}
Therefore,
\begin{equation}
h(s'_Z)
=
3h(s_Z)+O(\epsilon_Z).
\label{epsilonsZ}
\end{equation}
From Eqs.~(\ref{epsilonsX}) and (\ref{epsilonsZ}), the secure key rate $R'$ after the P-step is given by 
\begin{equation}
\begin{split}
R'
=
1-h(s'_X)-h(s'_Z)
&=
\frac94(1-h(s_X))
-3h(s_Z)
+O(\epsilon_X^4)
+O(\epsilon_Z)
\\
&=
\frac94(1-h(s_X)-h(s_Z))
-\frac34h(s_Z)
+O(\epsilon_X^4)
+O(\epsilon_Z)
\\
&=
\frac94R
-\frac34h(s_Z)
+O(\epsilon_X^4)
+O(\epsilon_Z).
\label{keyratePstepdisc2}
\end{split}
\end{equation}
From Eq.~(\ref{hsZ}), the negative contribution
\be
-\frac34 h(s_Z)
\ee
dominates the higher-order correction terms in Eq.~(\ref{keyratePstepdisc2}) under the assumptions
\be
\label{tinyassumption}
\epsilon_X,\epsilon_Z\ll1,
\qquad
-\epsilon_Z\log_2\epsilon_Z \gg \epsilon_X^4.
\ee
Therefore, asymptotically, \(R=0\) implies \(R'<0\), while \(R'=0\) implies \(R>0\). In other words, at least under these assumptions, one should not expect the P-step to improve the threshold.
Indeed, for example, when the initial \(Z\)- and \(X\)-error rates are both \(15\%\), after five B-steps, the value of $\epsilon_X$ is approximately \(4\times10^{-7}\), while $\epsilon_Z$ is approximately \(8\times10^{-25}\), which well satisfies the assumptions in Eq.~(\ref{tinyassumption}).

\section{Conclusion}
\label{Sec:Conclusion}
In this paper, we analyzed the asymptotic security of QCKA with tripartite GHZ states and two measurement bases by considering not only a single B-step but also multiple B-steps and their combination with P-steps.
For a single B-step, we first formulated the corresponding optimization problem and derived eight analytical solutions by exploiting the geometric structure of the objective function and the constraints.
We numerically confirmed that only one of these eight cases yields a positive key rate for a single B-step.
For this particular case, we extended the analysis to multiple B-steps by recursively applying the single-step analysis, and showed that the optimal solution of the optimization problem coincides with that of the single B-step case.
We further proved that this property remains valid even when multiple P-steps are applied after the B-steps.
Using these results, we numerically demonstrated that iterative application of the B-step enables the tolerable error threshold to exceed 20\%, representing a substantial improvement over the approximately 11\% threshold achievable without B-steps and the approximately 15\% threshold obtained with only a single B-step.

An important direction for future work is extending the present theory to more general settings.
First, in this work, we assumed that the state distributed to Alice, Bob, and Charlie is a qubit state (e.g., a single-photon state in optical implementations).
The present analysis is expected to remain valid even when multi-photon components are present, by employing the squashing method for two-basis QKD protocols~\cite{2008Beaudry, 2008Tsurumaru}.
Second, although our analysis considered asymptotic security under the IID assumption (collective attacks), extending the proof to finite-key security against general attacks remains an open problem.  Techniques such as the quantum de Finetti theorem~\cite{2007Renner} or the post-selection technique~\cite{2009Matthias} may be useful for this purpose, although careful treatment is required.
Finally, the most interesting open problem is extending the present analysis to QCKA with a larger number of parties.
As the number of parties increases, the number of optimization variables also grows, making it significantly more difficult to derive comprehensive analytical solutions, as already observed in Sec.~\ref{Sec:singleB}.
Nevertheless, our results imply that once a suitable solution for the single B-step case is identified, the analysis of multiple B-steps can become drastically simpler.
This suggests that a similar approach may also be applicable to multipartite QCKA with more parties.
Such an extension would be important for realizing high-error-threshold QCKA in general multipartite settings, and remains an important subject for future work.

\section*{Acknowledgement}
We thank Yoshihide Tonomura, Koji Azuma, Akihiro Mizutani and Devashish Tupkary for helpful discussions. 
We are grateful for support from the MIC R\&D of ICT Priority Technology (Grant No. JPMI00316) of the Ministry of Internal Affairs and Communications of Japan.

\appendix

\section{Analytical solutions for the optimization problem Eq.~(\ref{problemTpyramid})}
\label{Appe:othercases}

In this appendix, we derive analytical solutions to the optimization problem in Eq.~(\ref{problemTpyramid}) for the seven non-corner cases, complementing the corner case discussed in Sec.~\ref{subsection:analyticalsolutions}. As discussed in the main text, these solutions are obtained under the three constraints Eqs.~(\ref{zero2}), (\ref{zero4}), and (\ref{zero6}). The remaining constraints, Eqs.~(\ref{zero1}), (\ref{zero3}), (\ref{zero5}), (\ref{zero7}), and (\ref{zero8}), are verified in Appendix~\ref{Appe:constraints}, thereby establishing that the solutions presented here are also solutions to the original optimization problem in Eq.~(\ref{maxs1problem}). Among the eight cases summarized in Table~\ref{Table:solutions}, Case 1 corresponds to the corner solution discussed in the main text. Cases 2--4 correspond to edge solutions, Cases 5--7 to face solutions, and Case 8 to an interior solution of the pyramid.

\subsection{ Case 2: Edge $e_1$}
\label{Sec:e_1}

We consider the case where the worst-case point lies on the edge of $T_{\rm pyramid}$,
defined by the unit vector $\bm{e}_1$.
First, we define the conditions in which this is the worst-case point. This is equivalent to
the condition when the point $\bm{P}$ is closest to the considered edge out of all
other cases. 
In other words, the projection of $\bm P$ onto the line spanned by
$\bm e_1$ lies on the corresponding edge of $T_{\rm pyramid}$,
while the displacement from this projection to $\bm P$
belongs to ${\rm cone}(\bm e^2,\bm e^3)$.
In the coordinate space where $\bm{S}$ is the origin,
this condition can be defined in terms of the aforementioned vectors as
\begin{equation}
    \exists \alpha<0, \beta\geq0, \gamma\geq0~~{\rm s.t.}~~ \bm{SP} = \alpha\bm{e}_1 +\beta\bm{e}^2 + \gamma\bm{e}^3,
\end{equation}
in which the analytical solution of the worst point is simply the
vector along $\bm{e}_1$,
\begin{equation}
    \bm{SP'} = \alpha\bm{e}_1.
\end{equation}
$\alpha$ can be obtained by utilizing Eq.~(\ref{e_1}),
\begin{equation}
\begin{split}
    \bm{e}_1 \cdot \bm{SP} &= \alpha\bm{e}_1\cdot\bm{e}_1 + \beta\bm{e}_1\cdot\bm{e}^2 + \gamma\bm{e}_1\cdot\bm{e}^3\\
    &= \alpha |\bm{e}_1|^2 = \alpha\\
    \therefore \alpha &= v_1 = \frac{1}{\sqrt{2}} (1 -2s_1 -s_3 -s_4).\\
\end{split}
\end{equation}
$\beta$, $\gamma$ can be obtained in a similar manner,
\begin{equation}
\begin{split}
    \bm{e}_2 \cdot \bm{SP} &= \alpha\bm{e}_2\cdot\bm{e}_1 + \beta\bm{e}_2\cdot\bm{e}^2 + \gamma\bm{e}_2\cdot\bm{e}^3\\
    &= \frac{1}{2}\alpha+\beta\\
    \therefore \beta &=  v_2 -\frac{1}{2}v_1 = \frac{1}{2\sqrt{2}}(1-2s_1-2s_2+s_3-s_4),\\
\end{split}
\end{equation}
and
\begin{equation}
\begin{split}
    \bm{e}_3 \cdot \bm{SP} &= \alpha\bm{e}_3\cdot\bm{e}_1 + \beta\bm{e}_3\cdot\bm{e}^2 + \gamma\bm{e}_3\cdot\bm{e}^3\\
    &= \frac{1}{2}\alpha+\gamma\\
    \therefore \gamma &=  v_3 -\frac{1}{2}v_1 = \frac{1}{2\sqrt{2}}(1-2s_1-2s_2-s_3+s_4).\\
\end{split}
\end{equation}
In this case, the coordinate of the worst case point $\bm{SP'}$ is
\begin{equation}
\begin{split}
    \bm{SP'} &= \alpha\bm{e}_1 = v_1\cdot\bm{e}_1\\
    &= \frac{1}{\sqrt{2}}(1-2s_1-s_3-s_4) \cdot \frac{1}{\sqrt{2}}
    \begin{pmatrix}
        0\\
        1\\
        1\\
    \end{pmatrix}\\
    &= \frac{1}{2}
    \begin{pmatrix}
        0\\
        1 -2s_1 -s_3 -s_4\\
        1 -2s_1 -s_3 -s_4\\
    \end{pmatrix}.
\end{split}
\end{equation}
Finally, returning to the original coordinate system where $(0,0,0)^\mathrm{T}$ is the origin, the
expression of point $\bm{P'}$ is obtained as
\begin{equation}
\begin{split}
    \bm{P'} &= \bm{SP'} + \bm{S}
    =\frac{1}{2}
    \begin{pmatrix}
        0\\
        1 -2s_1 -s_3 -s_4\\
        1 -2s_1 -s_3 -s_4\\
    \end{pmatrix}
    +
    \begin{pmatrix}
        s_1+s_2\\
        s_1+s_3\\
        s_1+s_4\\
    \end{pmatrix}\\
    &=
     \begin{pmatrix}
        s_1+s_2\\
        \frac{1}{2}(1+s_3-s_4)\\
        \frac{1}{2}(1-s_3+s_4)\\
    \end{pmatrix}.
\end{split}
\end{equation}

To summarize, the worst case $s_5, s_6, s_7$ defined by the point on the edge $\bm e_1$ of $T_{\rm pyramid}$ 
is analytically obtained as
\begin{equation}
    \label{solutione_1}
    \bm{P'}=
     \begin{pmatrix}
        s_1+s_2\\
        \frac{1}{2}(1+s_3-s_4)\\
        \frac{1}{2}(1-s_3+s_4)\\
    \end{pmatrix},
\end{equation}
given that $\alpha<0, \beta\geq0, \gamma\geq0$ is satisfied, or in terms of $s_i$,
\begin{align}
    1-2s_1-s_3-s_4&<0,\label{edgecond1}\\
    1-2s_1-2s_2+s_3-s_4&\geq 0,\label{edgecond2}\\
    1-2s_1-2s_2-s_3+s_4&\geq 0,\label{edgecond3}
\end{align}
namely, 
\begin{align}
1 &< 2s_1+s_3+s_4,\\
|s_3-s_4| &\le 1-2s_1-2s_2.
\end{align}

\subsection{ Case 3: Edge $e_2$}
\label{Sec:e_2}
The analysis for Edge $\bm e_2$ is completely analogous to that for Edge $\bm e_1$.
By exchanging
\[
(\bm e_1, \bm e^2)\leftrightarrow (\bm e_2, \bm e^1),
\]
the worst-case point is obtained as
\begin{equation}
\label{solutione_2}
\bm{P'}=
\begin{pmatrix}
\frac{1}{2}(1+s_2-s_4)\\
s_1+s_3\\
\frac{1}{2}(1-s_2+s_4)
\end{pmatrix},
\end{equation}
provided that
\begin{align} 
1 -2s_1 -s_2 -s_4&<0,\label{e2edgecond1}\\ 
1-2s_1+s_2-2s_3-s_4&\geq 0,\label{e2edgecond2}\\ 
1-2s_1-s_2-2s_3+s_4&\geq 0,\label{e2edgecond3} 
\end{align}
namely,
\begin{align}
1 &< 2s_1+s_2+s_4,\\
|s_2-s_4| &\le 1-2s_1-2s_3.
\end{align}

\subsection{ Case 4: Edge $e_3$}
\label{Sec:e_3}

The analysis for Edge $\bm e_3$ is also completely analogous to that for Edge $\bm e_1$.
By exchanging
\[
(\bm e_1, \bm e^3)\leftrightarrow (\bm e_3, \bm e^1),
\]
the worst-case point is obtained as
\begin{equation}
    \label{solutione_3}
    \bm{P'}=
     \begin{pmatrix}
        \frac{1}{2}(1+s_2-s_3)\\
        \frac{1}{2}(1-s_2+s_3)\\
        s_1+s_4\\
    \end{pmatrix},
\end{equation}
provided that 
\begin{align}
    1 -2s_1 -s_2 -s_3&<0,\label{e3edgecond1}\\
    1-2s_1+s_2-s_3-2s_4&\geq 0,\label{e3edgecond2}\\
    1-2s_1-s_2+s_3-2s_4&\geq 0,\label{e3edgecond3}
\end{align}
namely,
\begin{align}
1 &< 2s_1+s_2+s_3,\\
|s_2-s_3| &\le 1-2s_1-2s_4.
\end{align}

\subsection{ Case 5: Face $e^1$}
\label{Sec:e^1}

We consider the case where the worst-case point lies on the face defined by the
normal $\bm{e}^1$.
First, we define the conditions in which this is the worst case point. This is equivalent to
the condition when the point $\bm{P}$ is closest to the considered face out of all
other cases. In other words, $\bm P$ lies above the considered face, and its orthogonal projection onto that face lies in the interior of the face.
This can be defined in terms of the aforementioned vectors as
\begin{equation}
    \exists \alpha\geq0, \beta<0, \gamma<0~~{\rm s.t.}~~ \bm{SP} = \alpha\bm{e}^1 +\beta\bm{e}_2 + \gamma\bm{e}_3.
\end{equation}
The analytical solution of the worst case point $\bm P'$ is
given by its orthogonal projection onto that face. In the coordinate space where $\bm{S}$ is the origin, this point is given as
$\alpha=0$, or
\begin{equation}
    \bm{SP'} = \beta\bm{e}_2 + \gamma\bm{e}_3.
\end{equation}
The coefficients $\alpha, \beta, \gamma$ can be expressed in terms of $s_i$ as follows.
First, $\alpha$ can be obtained by utilizing Eq.~(\ref{e^1}),
\begin{equation}
\begin{split}
    \bm{e}^1 \cdot \bm{SP} &= \alpha\bm{e}^1\cdot\bm{e}^1 + \beta\bm{e}^1\cdot\bm{e}_2 + \gamma\bm{e}^1\cdot\bm{e}_3\\
    &= \alpha |\bm{e}^1|^2 = \frac{3}{2} \alpha\\
    \therefore \alpha &= \frac{2}{3} \bm{e}^1\cdot\bm{SP} = \frac{2}{3} \nu^1 =\frac{\sqrt{2}}{3}(\frac{1}{2} -s_1 +s_2 -s_3 -s_4).\\
\end{split}
\end{equation}
$\beta$, $\gamma$ can be obtained in a similar manner,
\begin{equation}
\begin{split}
    \bm{e}^2 \cdot \bm{SP} &= \alpha\bm{e}^2\cdot\bm{e}^1 + \beta\bm{e}^2\cdot\bm{e}_2 + \gamma\bm{e}^2\cdot\bm{e}_3\\
    & = \beta-\frac{1}{2}\alpha\\
    \therefore \beta &=  \frac{1}{2}\alpha + v^2 = \frac{\sqrt{2}}{3}(1-2s_1-s_2+s_3-2s_4),\\
\end{split}
\end{equation}
and
\begin{equation}
\begin{split}
    \bm{e}^3 \cdot \bm{SP} &= \alpha\bm{e}^3\cdot\bm{e}^1 + \beta\bm{e}^3\cdot\bm{e}_2 + \gamma\bm{e}^3\cdot\bm{e}_3\\
    &= \gamma - \frac{1}{2}\alpha\\
    \therefore \gamma &= \frac{1}{2}\alpha +v^3  = \frac{\sqrt{2}}{3}(1-2s_1-s_2-2s_3+s_4).\\
\end{split}
\end{equation}
Then, the coordinate of the worst case point $\bm{SP'}$ is
\begin{equation}
\begin{split}
    \bm{SP'} &= \beta\bm{e}_2 + \gamma\bm{e}_3\\
    &= \frac{\sqrt{2}}{3}(1-2s_1-s_2+s_3-2s_4) \cdot \frac{1}{\sqrt{2}}
    \begin{pmatrix}
        1\\
        0\\
        1\\
    \end{pmatrix}
    +\frac{\sqrt{2}}{3}(1-2s_1-s_2-2s_3+s_4) \cdot \frac{1}{\sqrt{2}}
    \begin{pmatrix}
        1\\
        1\\
        0\\
    \end{pmatrix}\\
    &= \frac{1}{3}
    \begin{pmatrix}
        2 -4s_1 -2s_2 -s_3 -s_4\\
        1 -2s_1 -s_2 -2s_3 +s_4\\
        1 -2s_1 -s_2 +s_3 -2s_4\\
    \end{pmatrix}.
\end{split}
\end{equation}
Finally, returning to the original coordinate system where $(0,0,0)^\mathrm{T}$ is the origin, the
expression of point $\bm{P'}$ is obtained as

\begin{equation}
\begin{split}
    \bm{P'} &= \bm{SP'} + \bm{S}
    = \frac{1}{3}
    \begin{pmatrix}
        2 -4s_1 -2s_2 -s_3 -s_4\\
        1 -2s_1 -s_2 -2s_3 +s_4\\
        1 -2s_1 -s_2 +s_3 -2s_4\\
    \end{pmatrix}
    +
    \begin{pmatrix}
        s_1+s_2\\
        s_1+s_3\\
        s_1+s_4\\
    \end{pmatrix}\\
    &= \frac{1}{3}
     \begin{pmatrix}
        2 -s_1 +s_2 -s_3 -s_4\\
        1 +s_1 -s_2 +s_3 +s_4\\
        1 +s_1 -s_2 +s_3 +s_4\\
    \end{pmatrix}.
\end{split}
\end{equation}

In summary, the worst case $s_5, s_6, s_7$ defined by the point on the face of $T_{\rm pyramid}$ is analytically obtained as
\begin{equation}
\label{solutione^1}
    \bm{P'}
    = \frac{1}{3}
     \begin{pmatrix}
        2 -s_1 +s_2 -s_3 -s_4\\
        1 +s_1 -s_2 +s_3 +s_4\\
        1 +s_1 -s_2 +s_3 +s_4\\
    \end{pmatrix},
\end{equation}
given that $\alpha\geq0, \beta<0, \gamma<0$ is satisfied, or in terms of $s_i$,
\begin{align}
    \frac{1}{2} -s_1 +s_2 -s_3 -s_4 &\geq 0,\label{surfacecond1}\\
    1-2s_1-s_2+s_3-2s_4 &< 0,\label{surfacecond2}\\
    1-2s_1-s_2-2s_3+s_4 &< 0,\label{surfacecond3}
\end{align}
namely,
\begin{align}
s_1-s_2+s_3+s_4 &\le \frac{1}{2},\\
1-2s_1-s_2 &< \min\{(2s_3-s_4),(-s_3+2s_4)\}. 
\end{align}

\subsection{ Case 6: Face $e^2$}
\label{Sec:e^2}

The analysis for Face $\bm e^2$ is completely analogous to that for Face $\bm e^1$.
By exchanging
\[
(\bm e^1, \bm e_2)\leftrightarrow (\bm e^2, \bm e_1),
\]
the worst-case point is obtained as
\begin{equation}
    \label{solutione^2}
    \bm{P'}
    = \frac{1}{3}
     \begin{pmatrix}
        1+s_1+s_2-s_3+s_4\\
        2 -s_1 -s_2 +s_3 -s_4\\
        1+s_1+s_2-s_3+s_4\\
    \end{pmatrix},
\end{equation}
provided that 
\begin{align}
    \frac{1}{2} -s_1 -s_2 +s_3 -s_4 &\geq 0,\label{e2surfacecond1}\\
    1-2s_1+s_2-s_3-2s_4 &< 0,\label{e2surfacecond2}\\
    1-2s_1-2s_2-s_3+s_4 &< 0,\label{e2surfacecond3}
\end{align}
namely, 
\begin{align}
s_1+s_2-s_3+s_4 &\le \frac{1}{2},\\
1-2s_1-s_3 &< \min\{(-s_2+2s_4),(2s_2-s_4)\}. 
\end{align}

\subsection{ Case 7: Face $e^3$}
\label{Sec:e^3}

The analysis for Face $\bm e^3$ is also completely analogous to that for Face $\bm e^1$.
By exchanging
\[
(\bm e^1, \bm e_3)\leftrightarrow (\bm e^3, \bm e_1),
\]
the worst-case point is obtained as
\begin{equation}
    \label{solutione^3}
    \bm{P'}
    = \frac{1}{3}
     \begin{pmatrix}
        1+s_1+s_2+s_3-s_4\\
        1+s_1+s_2+s_3-s_4\\
        2-s_1-s_2-s_3+s_4\\
    \end{pmatrix},
\end{equation}
provided that
\begin{align}
    \frac{1}{2} -s_1 -s_2 -s_3 +s_4 &\geq 0,\label{e3surfacecond1}\\
    1-2s_1+s_2-2s_3-s_4 &< 0,\label{e3surfacecond2}\\
    1-2s_1-2s_2+s_3-s_4 &< 0,\label{e3surfacecond3}
\end{align}
namely,
\begin{align}
s_1+s_2+s_3-s_4 &\le \frac{1}{2},\\
1-2s_1-s_4 &< \min\{(-s_2+2s_3),(2s_2-s_3)\}.
\end{align}

\subsection{ Case 8: Within the pyramid  }
\label{Sec:within}

We consider the case where the worst-case point lies within $T_{\rm pyramid}$. In this case, the point $\bm{P} = (1/2,1/2,1/2)^{\mathrm{T}}$
becomes the realizable worst case.
In the coordinate space where $\bm{S}$ is the origin,
this is defined in terms of the aforementioned vectors as
\begin{equation}
    \exists v^1<0, v^2<0, v^3<0~~{\rm s.t.}~~ \bm{SP} = v^1\bm{e}_1 +v^2\bm{e}_2 + v^3\bm{e}_3.
\end{equation}
Then, the worst-case point is $\bm{P}$,
\begin{equation}
    \label{solutionwithin}
    \bm{P'} = \bm{P} =
    \begin{pmatrix}
        1/2\\
        1/2\\
        1/2\\
    \end{pmatrix},
\end{equation}
given that $v^1<0, v^2<0, v^3<0$, or in terms of $s_i$,

\begin{align}
    \frac{1}{2}-s_1+s_2-s_3-s_4&<0,\label{incond1}\\
    \frac{1}{2}-s_1-s_2+s_3-s_4&<0,\label{incond2}\\
    \frac{1}{2}-s_1-s_2-s_3+s_4&<0,\label{incond3}
\end{align}
namely, 
\be
\frac{1}{2}-s_1 < \min\{(-s_2+s_3+s_4), (s_2-s_3+s_4), (s_2+s_3-s_4)\}.
\ee

\subsection{Summary of the eight cases}
In summary, we have shown that there exist eight analytical solutions, depending on the conditions satisfied by $s_1, s_2, s_3, s_4$. These solutions are derived so as to satisfy the three constraints, Eqs.~(\ref{zero2}), (\ref{zero4}) and (\ref{zero6}), among the eight physical constraints. In Appendix~\ref{Appe:constraints}, we show that they also satisfy the other five physical constraints. 
The analytical solutions and their corresponding conditions are summarized in Table~\ref{Table:solutions}. 
Finally, we note that, from a more unified perspective, the solutions and conditions can be described as follows:

{\it 
We define $U=\{1,2,3\}$ and denote its power set by $2^U$.
Let $V\in 2^U$ be the subset satisfying
\begin{equation}
    \forall i\in V,\ c^i<0,
    \qquad
    \forall j\in U\setminus V,\ c_j\ge 0,
\end{equation}
where
\begin{equation}
    \bm{SP}
    =
    \sum_{i\in V} c^i\bm e_i
    +
    \sum_{j\in U\setminus V} c_j\bm e^j .
\end{equation}
Then $V$ corresponds to one of the eight cases, and the worst point $\bm P'$ is given by
\begin{equation}
    \bm{SP'}
    =
    \sum_{i\in V} c^i \bm e_i .
\end{equation}
}

\section{Proof that worst-case solutions satisfy physical constraints Eqs.~(\ref{zero1}), (\ref{zero3}), (\ref{zero5}), (\ref{zero7}) and (\ref{zero8})} \label{Appe:constraints}
In this appendix, we prove that the eight analytical solutions obtained under the constraints Eqs.~(\ref{zero2}), (\ref{zero4}) and (\ref{zero6}) in Sec.~\ref{subsection:analyticalsolutions} and Appendix~\ref{Appe:othercases} satisfy the remaining physical constraints given in Eqs.~(\ref{zero1}), (\ref{zero3}), (\ref{zero5}), (\ref{zero7}) and (\ref{zero8}). 
Note that the point $\bm{P}=(1/2,1/2,1/2)$ satisfies all of these constraints.  
However, it is not obvious that the point $\bm{P'}$, defined as the closest point to $\bm{P}$ within the pyramid $T_{\rm pyramid}$ specified by Eqs.~(\ref{zero2}), (\ref{zero4}) and (\ref{zero6}), also satisfies these constraints. Therefore, we verify this explicitly for each of the eight cases.

\subsection{Case 1: Corner}
\label{satisfyingotherconstraintscorner}

We consider the case where the solution is given by the corner of the pyramid $\bm{S}$, corresponding to the case discussed in Sec.~\ref{subsection:analyticalsolutions}.  

For Eq.~(\ref{zero1}), substitution of the solution in Eq.~(\ref{solutioncorner}) to the left-hand side gives
\begin{equation}
\label{corners4s6minuss7}
\begin{split}
    s_5+s_6-s_7 &= s_1 +s_2 +s_1 +s_3 -s_1 -s_4\\
    &= s_1 +s_2 +s_3 -s_4.
\end{split}
\end{equation}
Then, Eq.~(\ref{zero1}) can be written in terms of $s_1, s_2, s_3, s_4$ as
\begin{equation}
\label{cornerzero1ins1235}
\begin{split}
    &s_1-s_2-s_3+s_4 \leq s_1 +s_2 +s_3 -s_4\\
    &\Longleftrightarrow s_1 -(s_2+s_3-s_4) \leq s_1 +(s_2 +s_3 -s_4),\\
\end{split}
\end{equation}
which is equivalent to the relation shown in Eq.~(\ref{bound1}) and is always satisfied
given that Eq.~(\ref{bound1_s1235}) is satisfied.

For Eq.~(\ref{zero3}), substitution of the solution to the left-hand side gives
\begin{equation}
\label{corners4minuss6s7}
\begin{split}
    s_5-s_6+s_7 &= s_1 +s_2 -s_1 -s_3 +s_1 +s_4\\
    &= s_1 +s_2 -s_3 +s_4.
\end{split}
\end{equation}
Then, Eq.~(\ref{zero3}) can be written in terms of $s_1, s_2, s_3, s_4$ as
\begin{equation}
\label{cornerzero3ins1235}
\begin{split}
    &s_1-s_2+s_3-s_4 \leq s_1 +s_2 -s_3 +s_4\\
    &\Longleftrightarrow s_1 -(s_2-s_3+s_4) \leq s_1 +(s_2 -s_3 +s_4),\\
\end{split}
\end{equation}
which is equivalent to the relation shown in Eq.~(\ref{bound2}) and is always satisfied
given that Eq.~(\ref{bound2_s1235}) is satisfied.

For Eq.~(\ref{zero5}), substitution of the solution to the left-hand side gives
\begin{equation}
\label{cornerminuss4s6s7}
\begin{split}
    -s_5+s_6+s_7 &= -s_1 -s_2 +s_1 +s_3 +s_1 +s_4\\
    &= s_1 -s_2 +s_3 +s_4.
\end{split}
\end{equation}
Then, Eq.~(\ref{zero5}) can be written in terms of $s_1, s_2, s_3, s_4$ as
\begin{equation}
\label{cornerzero5ins1235}
\begin{split}
    &s_1+s_2-s_3-s_4 \leq s_1 -s_2 +s_3 +s_4\\
    &\Longleftrightarrow s_1 -(-s_2+s_3+s_4) \leq s_1 +(-s_2 +s_3 +s_4),\\
\end{split}
\end{equation}
which is equivalent to the relation shown in Eq.~(\ref{bound3}) and is always satisfied
given that Eq.~(\ref{bound3_s1235}) is satisfied.

The remaining constraints are Eq.~(\ref{zero7}) and Eq.~(\ref{zero8}).
Substitution of the solution to $s_5 +s_6 +s_7$ gives
\begin{equation}
\label{subcorner1}
\begin{split}
    s_5+s_6+s_7 &= s_1 +s_2 +s_1 +s_3 +s_1 +s_4\\
    &= 3s_1 +s_2 +s_3 +s_4.
\end{split}
\end{equation}

For Eq.~(\ref{zero7}), substitution of the solution gives the constraint to satisfy in terms of
$s_1, s_2, s_3, s_4$ as
\begin{equation}
\label{corners1-5constraints1}
\begin{split}
    &-s_1 + s_2 + s_3 + s_4 \leq 3s_1 +s_2 +s_3 +s_4\\
    &\Longleftrightarrow 0 \leq s_1.\\
\end{split}
\end{equation}
which is satisfied by definition.

For Eq.~(\ref{zero8}), substitution of the solution to the left-hand side gives
\begin{equation}
\label{corners1-5constraints2}
\begin{split}
    &3s_1 +s_2 +s_3 +s_4\leq 4 - (s_1+s_2+s_3+s_4)\\
    &\Longleftrightarrow(2s_1+s_2+s_3)+s_4\leq 2.\\
\end{split}
\end{equation}
Given the bound $2s_1+s_2+s_3\leq1$ from Eq.~(\ref{cornercond3}), $(2s_1+s_2+s_3)+s_4\leq1+s_4<2$ for $s_4<1/2$.

Therefore, all physical constraints
Eqs.~(\ref{zero1})--(\ref{zero8}) are always satisfied for the analytical solution in the regime that we consider.

\subsection{Case 2-4: Edge $e_1$, $e_2$, $e_3$}

We first consider the case where the solution is on the edge specified by the vector $\bm{e}_1$, corresponding to Case 2 discussed in Sec.~\ref{Sec:e_1}.  

For Eq.~(\ref{zero1}), substitution of the solution in Eq.~(\ref{solutione_1}) to the left-hand side gives
\begin{equation}
\label{edges4s6minuss7}
\begin{split}
    s_5+s_6-s_7 &= s_1 + s_2 +\frac{1}{2}(1+s_3-s_4)  -\frac{1}{2}(1-s_3+s_4)\\
    &= s_1 +s_2 +s_3 -s_4.
\end{split}
\end{equation}
Then, Eq.~(\ref{zero1}) can be written in terms of $s_1, s_2, s_3, s_4$ as
\begin{equation}
\label{edgezero1ins1235}
\begin{split}
    &s_1-s_2-s_3+s_4 \leq s_1 +s_2 +s_3 -s_4\\
    &\Longleftrightarrow s_1 -(s_2+s_3-s_4) \leq s_1 +(s_2 +s_3 -s_4),\\
\end{split}
\end{equation}
which is equivalent to the relation shown in Eq.~(\ref{bound1}) and is always satisfied
given that Eq.~(\ref{bound1_s1235}) is satisfied.

For Eq.~(\ref{zero3}), substitution of the solution to the left-hand side gives

\begin{equation}
\label{edges4minuss6s7}
\begin{split}
    s_5-s_6+s_7 &= s_1 + s_2 -\frac{1}{2}(1+s_3-s_4)  +\frac{1}{2}(1-s_3+s_4)\\
    &= s_1 +s_2 -s_3 +s_4.
\end{split}
\end{equation}
Then, Eq.~(\ref{zero3}) can be written in terms of $s_1, s_2, s_3, s_4$ as
\begin{equation}
\label{edgezero3ins1235}
\begin{split}
    &s_1-s_2+s_3-s_4 \leq s_1 +s_2 -s_3 +s_4\\
    &\Longleftrightarrow s_1 -(s_2-s_3+s_4) \leq s_1 +(s_2 -s_3 +s_4),\\
\end{split}
\end{equation}
which is equivalent to the relation shown in Eq.~(\ref{bound2}) and is always satisfied
given that Eq.~(\ref{bound2_s1235}) is satisfied.

For Eq.~(\ref{zero5}), substitution of the solution to the left-hand side gives
\begin{equation}
\label{edgeminuss4s6s7}
\begin{split}
    -s_5+s_6+s_7 &= -s_1 - s_2 +\frac{1}{2}(1+s_3-s_4)  +\frac{1}{2}(1-s_3+s_4)\\
    &= 1-s_1 -s_2.
\end{split}
\end{equation}
Then, Eq.~(\ref{zero5}) can be written in terms of $s_1, s_2, s_3, s_4$ as
\begin{equation}
\label{edgezero5ins1235}
\begin{split}
    &s_1 +s_2 -s_3 -s_4 \leq 1 -s_1 -s_2\\
    &\Longleftrightarrow 2s_1 +2s_2 -s_3 -s_4 \leq 1\\
    &\Longleftrightarrow (2s_1 +2s_2 -s_3 +s_4) -2s_4 \leq 1.\\
\end{split}
\end{equation}
From Eq.~(\ref{edgecond2}), the bracketed expression in the last line
of Eq.~(\ref{edgezero5ins1235}) is bounded by $2s_1 +2s_2 -s_3 +s_4\leq1$, so  
Eq.~(\ref{edgezero5ins1235}) is always satisfied for $0\leq s_4$.

The remaining constraints are Eq.~(\ref{zero7}) and Eq.~(\ref{zero8}).
Substitution of the solution to $s_5 +s_6 +s_7$ gives
\begin{equation}
\label{subedge1}
\begin{split}
    s_5+s_6+s_7 &= +s_1 + s_2 +\frac{1}{2}(1+s_3-s_4)  +\frac{1}{2}(1-s_3+s_4)\\
    &=1+s_1+s_2.
\end{split}
\end{equation}

For Eq.~(\ref{zero7}), substitution of Eq.~(\ref{subedge1}) gives the constraint to be satisfied in terms of
$s_1, s_2, s_3, s_4$ as
\begin{equation}
\label{edges1-5constraints1}
\begin{split}
    &-s_1 + s_2 + s_3 + s_4 \leq 1+s_1+s_2\\
    &\Longleftrightarrow -2s_1 + s_3 + s_4 \leq 1.\\
\end{split}
\end{equation}
In the region $0\leq s_1, s_3, s_4<1/2$, we have $-2s_1 + s_3 + s_4<1$, and thus 
the above equation is always satisfied.

For Eq.~(\ref{zero8}), substitution of the solution to the left-hand side gives
\begin{equation}
\begin{split}
    &1+s_1+s_2\leq 4 - (s_1+s_2+s_3+s_4)\\
    &\Longleftrightarrow2s_1+2s_2+s_3+s_4\leq 3\\
\end{split}
\end{equation}
In the region $0\leq s_1, s_2, s_3, s_4<1/2$, we have $2s_1+2s_2+s_3+s_4<3$, and thus 
the above equation is always satisfied.

Therefore, all physical constraints
Eqs.~(\ref{zero1})--(\ref{zero8}) are always satisfied for the analytical solution in the regime that we consider.

By symmetry of the constraints, the verification that the solutions on edges $\bm e_2$ (Case 3) and $\bm e_3$ (Case 4) satisfy all constraints reduces to that of $\bm e_1$ (Case 2) under the exchanges $s_2 \leftrightarrow s_3$ and $s_2 \leftrightarrow s_4$, respectively. Therefore, all constraints are satisfied in the same manner.

\subsection{Case 5-7: Face $e^1$, $e^2$, $e^3$} 
We consider the case where the solution is on the face specified by the vector $\bm{e}^1$, corresponding to Case 5 discussed in Sec.~\ref{Sec:e^1}. 

For Eq.~(\ref{zero1}), substitution of the solution in Eq.~(\ref{solutione^1}) to the left-hand side gives
\begin{equation}
\label{subs4s6minuss7}
\begin{split}
    s_5+s_6-s_7 &= \frac{1}{3} (2 -s_1 +s_2 -s_3 -s_4 +1 +s_1 -s_2 +s_3 +s_4 -1 -s_1 +s_2 -s_3 -s_4)\\
    &= \frac{1}{3}(2 -s_1 +s_2 -s_3 -s_4).
\end{split}
\end{equation}
Then, Eq.~(\ref{zero1}) can be written in terms of $s_1, s_2, s_3, s_4$ as
\begin{equation}
\label{zero1ins1235}
\begin{split}
    &s_1-s_2-s_3+s_4 \leq \frac{1}{3}(2 -s_1 +s_2 -s_3 -s_4)\\
    &\Longleftrightarrow 2s_1-2s_2-s_3+2s_4 \leq 1\\
    &\Longleftrightarrow (2s_1-2s_2+2s_3+2s_4) -3s_3 \leq 1.\\
\end{split}
\end{equation}
From Eq.~(\ref{surfacecond1}), the bracketed expression in the last line
of Eq.~(\ref{zero1ins1235}) is bounded by $2s_1-2s_2+2s_3+2s_4\leq1$, so  
Eq.~(\ref{zero1ins1235}) is always satisfied for $0\leq s_3$.

For Eq.~(\ref{zero3}), substitution of the solution to the left-hand side gives
the same expression as Eq.~(\ref{subs4s6minuss7}) as $s_6=s_7$ holds for the current solution,  
\begin{equation}
\label{subs4minuss67}
\begin{split}
    s_5-s_6+s_7 &= \frac{1}{3} (2 -s_1 +s_2 -s_3 -s_4 -1 -s_1 +s_2 -s_3 -s_4 +1 +s_1 -s_2 +s_3 +s_4)\\
    &= \frac{1}{3}(2 -s_1 +s_2 -s_3 -s_4).
\end{split}
\end{equation}
Then, Eq.~(\ref{zero3}) can be written in terms of $s_1, s_2, s_3, s_4$ as
\begin{equation}
\label{zero3ins1235}
\begin{split}
    &s_1-s_2+s_3-s_4 \leq \frac{1}{3}(2 -s_1 +s_2 -s_3 -s_4)\\
    &\Longleftrightarrow 2s_1-2s_2+2s_3-s_4 \leq 1.\\
    &\Longleftrightarrow (2s_1-2s_2+2s_3+2s_4) -3s_4 \leq 1.\\
\end{split}
\end{equation}
Similar to the previous argument, the bracketed expression in the last line
of Eq.~(\ref{zero3ins1235}) is bounded by $2s_1-2s_2+2s_3+2s_4\leq1$ from Eq.~(\ref{surfacecond1}),
so Eq.~(\ref{zero3ins1235}) is always satisfied for $0\leq s_4$.

For Eq.~(\ref{zero5}), substitution of the solution to the left-hand side gives
\begin{equation}
\label{minuss4s6s7}
\begin{split}
    -s_5+s_6+s_7 &= \frac{1}{3} (-2 +s_1 -s_2 +s_3 +s_4 +1 +s_1 -s_2 +s_3 +s_4 +1 +s_1 -s_2 +s_3 +s_4)\\
    &=s_1 -s_2 +s_3 +s_4.
\end{split}
\end{equation}
Then, Eq.~(\ref{zero5}) can be written in terms of $s_1, s_2, s_3, s_4$ as
\begin{equation}
\label{zero5ins1235}
\begin{split}
    &s_1 +s_2 -s_3 -s_4 \leq s_1 -s_2 +s_3 +s_4\\
    &\Longleftrightarrow s_1 -(-s_2 +s_3 +s_4) \leq s_1 + (-s_2 +s_3 +s_4),\\
\end{split}
\end{equation}
which is equivalent to the relation shown in Eq.~(\ref{bound3}) and is always satisfied
given that Eq.~(\ref{bound3_s1235}) is satisfied.

The remaining constraints are Eq.~(\ref{zero7}) and Eq.~(\ref{zero8}).
Substitution of the solution to $s_5 +s_6 +s_7$ is performed, yielding
\begin{equation}
\label{subsurface1}
\begin{split}
    s_5+s_6+s_7 &= \frac{1}{3} (2 -s_1 +s_2 -s_3 -s_4 +1 +s_1 -s_2 +s_3 +s_4 +1 +s_1 -s_2 +s_3 +s_4)\\
    &=\frac{1}{3} (4 +s_1 -s_2 +s_3 +s_4).
\end{split}
\end{equation}

For Eq.~(\ref{zero7}), substitution of Eq.~(\ref{subsurface1}) yields the constraint to satisfy in terms of
$s_1, s_2, s_3, s_4$ as
\begin{equation}
\label{s1-5constraints1}
\begin{split}
    &-s_1 + s_2 + s_3 + s_4 \leq \frac{1}{3} (4 +s_1 -s_2 +s_3 +s_4)\\
    &\Longleftrightarrow -4s_1 + 4s_2 + 2s_3 + 2s_4 \leq 4.\\
\end{split}
\end{equation}
In the region $0\leq s_1, s_2, s_3, s_4<1/2$, we have $-4s_1 + 4s_2 + 2s_3 + 2s_4<4$, and thus 
the above equation is always satisfied.

For Eq.~(\ref{zero8}), substitution of Eq.~(\ref{subsurface1}) yields the following constraint as
\begin{equation}
\begin{split}
    &\frac{1}{3} (4 +s_1 -s_2 +s_3 +s_4) \leq 4 - (s_1 + s_2 + s_3 + s_4)\\
    &\Longleftrightarrow 2s_1  + s_2 +2s_3 +2s_4 \leq 4,
\end{split}
\end{equation}
In the region of $0\leq s_1, s_2, s_3, s_4<1/2$, we have $2s_1 + s_2 +2s_3 +2s_4<7/2$, and thus
the above equation is always satisfied.

Therefore, all physical constraints
Eqs.~(\ref{zero1})--(\ref{zero8}) are always satisfied for the analytical solution in the regime that we consider.

By symmetry of the constraints, the verification that the solutions on faces $\bm e^2$ (Case 6) and $\bm e^3$ (Case 7) satisfy all constraints reduces to that of face $\bm e^1$ (Case 5) under the exchanges $s_2 \leftrightarrow s_3$ and $s_2 \leftrightarrow s_4$, respectively. Therefore, all constraints are satisfied in the same manner.

\subsection{Case 8: Within the pyramid}

We consider the case where the solution lies within the pyramid $T_{\rm pyramid}$ and is given by $\bm{P}$, corresponding to the case discussed in Sec.~\ref{Sec:within}.  

For Eq.~(\ref{zero1}), substitution of the solution in Eq.~(\ref{solutionwithin}) allows the constraint to be written
in terms of $s_1, s_2, s_3, s_4$ as
\begin{equation}
\label{inzero1ins1235}
    s_1-s_2-s_3+s_4 \leq \frac{1}{2}.
\end{equation}
From Eq.~(\ref{incond3}), we have $-s_1-s_2-s_3+s_4<-1/2$, 
which leads to 
\begin{equation}
    s_1-s_2-s_3+s_4 < -\frac{1}{2} + 2s_1, 
\end{equation}
so Eq.~(\ref{inzero1ins1235}) is always satisfied in the region $0\leq s_1<1/2$.
Similar discussions can be made for Eq.~(\ref{zero3}), Eq.~(\ref{zero5}).

For Eq.~(\ref{zero3}), substitution of the solution allows the constraint to be written
in terms of $s_1, s_2, s_3, s_4$ as
\begin{equation}
\label{inzero3ins1235}
    s_1-s_2+s_3-s_4 \leq \frac{1}{2}.
\end{equation}
From Eq.~(\ref{incond2}), we have $-s_1-s_2+s_3-s_4<-1/2$, 
which leads to 
\begin{equation}
    s_1-s_2+s_3-s_4 < -\frac{1}{2} + 2s_1, 
\end{equation}
so Eq.~(\ref{inzero3ins1235}) is always satisfied in the region $0\leq s_1<1/2$.

For Eq.~(\ref{zero5}), substitution of the solution allows the constraint to be written
in terms of $s_1, s_2, s_3, s_4$ as
\begin{equation}
\label{inzero5ins1235}
    s_1+s_2-s_3-s_4 \leq \frac{1}{2}.
\end{equation}
From Eq.~(\ref{incond1}), we have $-s_1+s_2-s_3-s_4<-1/2$, 
which leads to 
\begin{equation}
    s_1+s_2-s_3-s_4 < -\frac{1}{2} + 2s_1, 
\end{equation}
so Eq.~(\ref{inzero5ins1235}) is always satisfied in the region $0\leq s_1<1/2$.

The remaining constraints to check are Eq.~(\ref{zero7}) and Eq.~(\ref{zero8}).

For Eq.~(\ref{zero7}), substitution of the solution gives the constraint to satisfy in terms of
$s_1, s_2, s_3, s_4$ as
\begin{equation}
\label{ins1-5constraints1}
    -s_1 + s_2 + s_3 + s_4 \leq \frac{3}{2},
\end{equation}
which is always satisfied in the region $s_1, s_2, s_3, s_4<1/2$.

For Eq.~(\ref{zero8}), substitution of the solution gives the constraint to satisfy in terms of
$s_1, s_2, s_3, s_4$ as
\begin{equation}
\label{ins1-5constraints2}
\begin{split}
    &\frac{3}{2}\leq 4 - (s_1+s_2+s_3+s_4)\\
    &\Longleftrightarrow s_1+s_2+s_3+s_4\leq \frac{5}{2},\\
\end{split}
\end{equation}
which is always satisfied in the region $s_1, s_2, s_3, s_4<1/2$.
 
Therefore, all physical constraints
Eqs.~(\ref{zero1})--(\ref{zero8}) are always satisfied for the analytical solution in the regime that we consider.


\bibliography{bib_2way}
\bibliographystyle{apsrev4-2}

\end{document}